\documentclass[sigconf]{acmart}

\usepackage[noend,ruled,vlined,linesnumbered]{algorithm2e}
\SetProcSty{textnormal}
\SetArgSty{textnormal}
\SetFuncSty{textnormal}
\SetProgSty{textnormal}
\SetInd{0.0em}{0.6em}

\usepackage{enumitem}
\usepackage{graphicx}
\usepackage{subfigure} 
\usepackage{mathtools}
\usepackage{wrapfig}
\usepackage{amsfonts}
\usepackage{amsmath}
\usepackage{array}
\usepackage{natbib}
\usepackage{units}
\usepackage{xcolor}
\usepackage{xspace}
\newcommand{\mA}{\mathbf{A}}

\newcommand{\mC}{\mathbf{C}}

\newcommand{\mM}{\mathbf{M}}
\newcommand{\mB}{\mathbf{B}}

\newcommand{\vx}{\boldsymbol{x}}
\newcommand{\vb}{\boldsymbol{b}}
\newcommand{\ATA}{\textsc{AtA}\xspace}
\newcommand{\ATAP}{\textsc{AtA-D}\xspace}
\newcommand{\ATAS}{\textsc{AtA-S}\xspace}
\newcommand{\ATAnaive}{\textsc{AtANaive}\xspace}
\newcommand{\Strassen}{\textsc{FastStrassen}\xspace}
\newcommand{\TaskTree}{\mathcal{T}}

\newtheorem{proposition}{Proposition}[section]

\AtBeginDocument{%
  \providecommand\BibTeX{{%
    \normalfont B\kern-0.5em{\scshape i\kern-0.25em b}\kern-0.8em\TeX}}}

\begin{document}

\title{Efficiently Parallelizable Strassen-Based Multiplication of a Matrix by its Transpose}


\author{Viviana Arrigoni}
\authornote{Both authors contributed equally to this research.}
\email{arrigoni@di.uniroma1.it}
\affiliation{%
  \institution{Department of Computer Science,\\ Sapienza University of Rome}
 \city{Rome}
  \country{Italy}
}
\author{Filippo Maggioli}
\authornotemark[1]
\email{maggioli@di.uniroma1.it}
\affiliation{%
  \institution{Department of Computer Science,\\ Sapienza University of Rome}
 \city{Rome}
  \country{Italy}
}
\author{Annalisa Massini}
\email{massini@di.uniroma1.it}
\affiliation{%
  \institution{Department of Computer Science,\\ Sapienza University of Rome}
 \city{Rome}
  \country{Italy}
}
\author{Emanuele Rodolà}
\email{rodola@di.uniroma1.it}
\affiliation{%
  \institution{Department of Computer Science,\\ Sapienza University of Rome}
 \city{Rome}
  \country{Italy}
}



\begin{abstract}
The multiplication of a matrix by its transpose, $\mA^T\mA$, appears as an intermediate operation in the solution of a wide set of problems.
In this paper, we propose a new cache-oblivious algorithm (\ATA) for computing this product, based upon the classical Strassen algorithm as a sub-routine. In particular, we decrease the computational cost 
to $\nicefrac{2}{3}$ the time required by Strassen's algorithm, amounting to $\frac{14}{3}n^{\log_2 7}$ floating point operations.
\ATA works for generic rectangular matrices, and exploits the peculiar symmetry of the resulting product matrix for saving memory. In addition, we provide an extensive implementation study of \ATA in a shared memory system, and extend its applicability to a distributed environment.
To support our findings, we compare our algorithm with state-of-the-art  solutions specialized in the computation of $\mA^T\mA$. 
Our experiments highlight good scalability with respect to both the matrix size and the number of involved processes,
as well as favorable performance for both the parallel paradigms and the sequential implementation, when compared with other methods in the literature.
\end{abstract}




\maketitle


\section{Introduction}
Matrix multiplication is a fundamental operation in Linear Algebra and HPC, 
as it appears as an intermediate step in a wide set of problems.  
Many researchers have devoted their efforts to the algorithmic aspects of matrix multiplication, with the aim of improving  the computational cost of existing algorithms and to devise and implement new solutions for parallel architectures. 
Designing a distributed algorithm for matrix multiplication is a challenging task, due to the inherent dependence of the data scattered in the system's distributed memory, and due to the overhead due to the communication cost of assembling the resulting product matrix.

The product of a matrix by its transpose, $\mA^T\mA$ (as well as $\mA\mA^T$), is a particular matrix multiplication involved in several applications. For example, computing $\mA\mA^T$ is a straightforward, yet effective, method to check for orthogonality or to project vectors onto the space spanned by the columns of $\mA$. This product, in fact, is repeatedly computed in the Gram-Schmidt algorithm for vector basis orthogonalization, where $\mA$ is the progressively built projection matrix.
One way to solve the least squares problem of under and over determined linear systems $\mA\vx=\vb$, is to solve the associated system of normal equations, obtained by left-hand multiplying the original system by $\mA^T$, thus obtaining a square linear system $\mA^T\mA \vx = \mA^T\vb$.
Also, the Singular Value Decomposition (SVD) of a matrix $\mA$ can be computed by studying the
eigenproblem for $\mA^T\mA$ and $\mA\mA^T$. 
Furthermore, the product of a matrix by its transpose commonly arises in discrete exterior calculus and discrete differential geometry. One example is given by the discrete heat kernel $\mathbf{K}(t) = \mathbf{\Phi}\mathbf{E}(t)\mathbf{\Phi}^T$, with $\mathbf{E}(t) = \exp(-\mathbf{\Lambda}t)$ being a diagonal matrix, so that  $\mathbf{K}(t) = (\mathbf{\Phi}\mathbf{E}(t)^{\nicefrac{1}{2}})(\mathbf{\Phi}\mathbf{E}(t)^{\nicefrac{1}{2}})^T$ can be efficiently computed~\cite{zeng2012discrete}. Many other applications of the product $\mA^T\mA$ are described in~\cite{strang06}, together with its properties such as positive semi-definiteness.





In this work, we consider the multiplication between $\mA^T$ and $\mA$, where $\mA$ may have any size and shape. We rely on a recursive approach that, as described in~\cite{peise2017algorithm}, is virtually tuning free and avoids the significant tuning efforts that are required by iterative blocked algorithms to achieve near-optimal performance. Our contribution is threefold. 
\begin{itemize}[leftmargin=*]
    \item First, we introduce \ATA (Section~\ref{sec:ATA}), a cache-oblivious algorithm for computing  $\mA^T\mA$ that requires $\nicefrac{2}{3}n^{(\log_2 7)}+\nicefrac{1}{3}n^2$ multiplications. We exploit the self-similarity of the $\mA^T\mA$ product with its sub-problems and the Strassen's algorithm, that is recursively applied to possibly rectangular matrices, without introducing additional computational and space cost, deriving from dynamic peeling and padding, as in \cite{Huss-Super96,thottethodi1998tuning}. In contrast to~\cite{dumas2020fast}, our algorithm works on any algebraic field. 
    We prove that \ATA exhibits high efficiency for both memory and time, and show that it is efficiently implementable, as it does not hide large constant factors. 
    We also describe our implementation of Strassen's algorithm, and compare its performance with that of the Intel MKL BLAS \texttt{gemm} routine for matrix multiplication. 
    \item Second, we describe \ATAS, our multi-threaded implementation of \ATA for a shared memory system,  relying on OpenMP (Section~\ref{sec:ATAS}). A well-engineered scheduler that 
    assigns different tasks to each thread in such a way that computations 
    can be carried out in perfect parallelism by preventing memory collisions. Performance evaluation shows that our implementation outperforms the multi-threaded Intel MKL BLAS routines (e.g. \texttt{syrk} for symmetric rank-$K$ update) on large matrices, even on 
    Intel processors.
    \item Finally, we extend our approach 
    to distributed systems, leveraging the standard message-passing paradigm MPI. Our distributed algorithm \ATAP allows the distribution of the computational effort among a larger number of processes (Section~\ref{sec:ATAP}). This is particularly convenient on very large matrices. 
\end{itemize}






 
To validate the effectiveness of our algorithms, we study their performance by running a set of tests on dense matrices of variable size (Section~\ref{sec:ATAtest}). We analyse different metrics for evaluating the scalability of our parallel implementation, and compare our results with benchmark solutions for distributed systems. We run tests on a cluster of multi-core nodes endowed with $2\times 8$ core Intel Xeon E5-2630v3 processors, 2.4 Ghz, 4 GB RAM/core.



\section{Related work}
\label{sec:related}

Nowadays, matrix multiplication is still a hot topic in HPC and numerical algorithmics. 
In 1969, Strassen \cite{strassen1969gaussian} was the first to reduce the computational complexity of the standard matrix multiplication from $O(n^3)$ to $O(n^{\log_2 7})$. 
More recently, Coppersmith and Winograd \cite{coppersmith1987matrix} devised an algorithm for matrix multiplication running in $\sim O(n^{2.38})$ time. In the last decade, many have devoted their efforts to improve this limit (\cite{stothers2010complexity,williams2012multiplying,le2014powers}). 
These works make use of algebraic tensors that, despite the elegance of the resulting method, are still hardly used in practice as they come at the cost of very large hidden constants and frequent memory access. 

Several authors have designed hybrid algorithms, deploying Strassen’s multiplication in conjunction with conventional matrix multiplication, to overcome the overhead of Strassen's algorithm on small matrices, see, e.g., \cite{brent-TR70,brent-NM70,Higham-TOMS90,Huss-Super96,Benson-PPOPP15}.
Huss-Lederman {\it et al.} \cite{Huss-Super96} propose two techniques, known as dynamic peeling and static padding, in order to apply Strassen's algorithm to odd-sized matrices. Thottethodi {\it et al.} \cite{thottethodi1998tuning} propose two strategies to optimize memory efficiency in Strassen by minimizing padding and peeling operations. 
Many researchers have proposed a parallel implementation of Strassen's algorithm. 
In \cite{Luo-SAC95}, Luo and Drake explored Strassen-based parallel algorithms that use the communication patterns known for classical matrix multiplication. They considered using a classical 2D parallel algorithm and using Strassen locally and at the highest level. 
This approach is improved in \cite{Grayson-PPL95} by using a more efficient parallel matrix-multiplication algorithm running on a more communication-efficient machine. 
In \cite{d2007adaptive}, Strassen’s algorithm is extended to deal with rectangular and arbitrary-size matrices.
Their approach leverages on a suitable combination of Strassen’s with ATLAS and GotoBLAS. 
Other parallel approaches \cite{Desprez-04,Hunold-08,Song-PDCS06} have used more complex parallel schemes and communication patterns, 
and consider at most two steps of Strassen. 
In \cite{ballard2012communication}, a parallel algorithm based on Strassen's fast matrix multiplication, Communication - Avoiding Parallel Strassen (CAPS), is described.
The authors show that its complexity matches the communication lower bounds described in \cite{BDHS-SPAA11}. This work is extended in \cite{demmel2013communication} to handle rectangular matrices (CARMA). %
More recently, Kwasniewski {\it et al.} \cite{COSMA} proposed a near optimal algorithm for matrix multiplication that models the matrix multiplication dependencies by the red-blue pebble game \cite{10.1145/800076.802486} to derive an I/O optimal schedule, improving the performance of previous works. 

Both Strassen's algorithm and \ATA fall into the class of recursive blocked algorithms. The work in~\cite{elmroth2004recursive,kaagstrom2004management} proves the effectiveness of this kind of algorithms for dense Linear Algebra. The work in~\cite{eliahu2015frpa} introduces FRPA, an interface for implementing recursive problems in parallel that gets as an input the recursive problem, and handles parallelization and  auto-tuning automatically. 
Similarly to our approach, Charara \emph{et al.} \cite{charara2016redesigning} propose block recursive matrix multiplication and linear solver algorithms. They show how recursion enhances data reuse and concurrency in GPUs. Differently from the work presented in this paper, they specialize on triangular matrices. In~\cite{charara2019batched}, the authors also adapt this blocking strategy to handle batched operations on small matrix sizes (up to 256) to stress the register usage and maintain data locality. 
In~\cite{peise2017algorithm}, Elmar and Bientinesi introduce ReLAPACK, a collection of  recursive algorithms for dense Linear Algebra. While this work corroborates the recursive approach that we implement in our algorithms, it does not provide a routine specialized in the $\mA^T\mA$ product for general matrices. Instead, they propose a routine for the same multiplication only on triangular matrices. We highlight that the solutions proposed for the multiplication of a matrix by its transpose on triangular matrices (TRSYRK) is useful for many applications but cannot be applied on general matrices. 

Although much research has been devoted to optimizing the  implementation of parallel matrix multiplication, very few solutions have been proposed for the $\mA^T\mA$ multiplication.   
In \cite{dumas2020fast}, Dumas \emph{et al.} propose an algorithm for the product $\mA\mA^T$ whose computational complexity  is improved by a constant factor with respect to previously known reductions. This approach is applicable only to matrices lying in fields where skew-orthogonal matrices exist (e.g., $\mathbb{C}$ and finite fields of prime characteristics), which is not the case for $\mathbb{R}$ and $\mathbb{Q}$, that instead are important in many applications, such as the study of embedded systems, computational geometry and system simulations. 

Except for some sporadic attempts to implement a method for distributing in a balanced way the workload for matrix multiplication among processes with the MapReduce programming model \cite{qasem2017matrix,kadhum2017efficient}, the approach that we implement here for the distributed parallel model has barely been investigated. 

\section{\ATA}
\label{sec:ATA}
In this section, we describe our sequential recursive algorithm for the matrix multiplication $\mA^T\mA$, dubbed \textsc{AtA}, and we provide implementation details. We remark that our solution also works for the product $\mA\mA^T$. Yet, when row-major order is the default layout for array storage, 
the $\mA^T\mA$ multiplication is in practice harder to perform, as memory access is inherently column-wise, hence not cache friendly. Since \ATA includes calls to Strassen for generic matrix multiplications, we also outline a time and space efficient implementation for this algorithm.

\subsection{\ATA in detail}
Let $\mA\in \mathbb{R}^{m\times n}$ be a rectangular matrix. The idea behind \ATA is the following: at each recursive step, matrix $\mA$ is divided into four sub-matrices as follows: 
\begin{equation}
\label{eq:submatA}
\begin{array}{cc}
    \mA=
\begin{bmatrix}
       \mA_{1,1} & \mA_{1,2}  \\
    \mA_{2,1} & \mA_{2,2}
     \end{bmatrix} &  
     \begin{array}{cc}
         \mA_{1,1}=\mA_{0:m_1,0:n_1} \in \mathbb{R}^{m_1 \times n_1}  \\
          \mA_{1,2}=\mA_{0:m_1,n_1:n}\in \mathbb{R}^{m_1 \times n_2}\\
          \mA_{2,1}=\mA_{m_1:m,0:n_1}\in \mathbb{R}^{m_2 \times n_1}\\
          \mA_{2,1}=\mA_{m_1:m,n_1:n}\in \mathbb{R}^{m_2 \times n_2}
     \end{array}
\end{array}
\end{equation} 
being $m_1  :=\left \lfloor \frac{m}{2}\right \rfloor$, $m_2 :=\left \lceil \frac{m}{2}\right \rceil$,  $n_1 :=\left \lfloor \frac{n}{2}\right \rfloor$, $n_2 :=\left \lceil \frac{n}{2}\right \rceil$. We address to sub-matrices of a matrix $\mA$ as to indexed sub-blocks ($\mA_{i,j}$) or with line and column intervals ($\mA_{r_1:r_2,c_1:c_2}$). The product matrix $\mC = \mA^T \mA$ is also split into four sub-matrices, resulting in the following:
    

\noindent
\begin{equation}
\begin{aligned}
\label{eq: Cij}
 \mC_{1,1} = \mA_{1,1}^T\mA_{1,1} + \mA_{2,1}^T\mA_{2,1}\in \mathbb{R}^{n_1 \times n_1},\\
\mC_{1,2} = \mA_{1,1}^T\mA_{1,2} + \mA_{2,1}^T\mA_{2,2}\in \mathbb{R}^{n_1 \times n_2},\\
\mC_{2,1} = \mA_{1,2}^T\mA_{1,1} + \mA_{2,2}^T\mA_{2,1}\in \mathbb{R}^{n_2 \times n_1},\\
\mC_{2,2} = \mA_{1,2}^T\mA_{1,2} + \mA_{2,2}^T\mA_{2,2}\in \mathbb{R}^{n_2 \times n_2}.
\end{aligned}
\end{equation}
Both $\mC_{1,1}$ and $\mC_{2,2}$ consist of two addends that are, in turn, the left hand product of a matrix by its transpose. 
Hence, four recursive calls are employed to compute the  sub-products $\mA_{1,1}^T \mA_{1,1}$ and $\mA_{2,1}^T \mA_{2,1}$ to obtain $\mC_{1,1}$, and $\mA_{1,2}^T \mA_{1,2}$ and $\mA_{2,2}^T\mA_{2,2}$ to obtain $\mC_{2,2}$. 

Since for any matrix $\mA$ the product $\mA^T\mA$ is symmetric, at each recursive step only the lower triangular part of the product matrix is computed, $\text{low}(\mC_{i,i})$, $i = 1,2$.
As for component $\mC_{2,1}$, in order to compute its two terms in the sum, 
we implement the generalized Strassen's algorithm for non-square matrices. The sub-matrix $\mC_{1,2}$ is equal to $\mC_{2,1}^T$, and therefore must not be explicitly computed. 
In Algorithm \ref{alg: AtA} we provide the pseudo-code of \textsc{AtA}. The base case occurs when the number of entries of the sub-matrix fits in the cache. 
In that case, the multiplication is performed by the BLAS function for $\mA^T\mA$, \texttt{?syrk}, where the character \texttt{?} represents a generic data type in accordance with standard notation used in manuals,~\cite{intelMKL}. In Algorithm~\ref{alg: AtA}, we also sketch our implementation of Strassen: before the actual recursive Strassen algorithm is called (\textsc{Strassen}), in \textsc{FastStrassen} we conveniently prepare an environment for memory efficiency by pre-allocating the memory for Strassen's algorithm, as explained in Section~\ref{sec:ATAS(n)}. The reduced number of multiplications in Strassen's algorithm is achieved by computing more matrix additions. In our implementation of \textsc{Strassen}, matrix additions are performed by calling the BLAS routine \texttt{?axpy} (for the vector addition $\boldsymbol{y} = \alpha \vx + \boldsymbol{y}$). The base-case condition in \textsc{Strassen} is analogous to the one of \ATA. When the base-case condition holds, we call the BLAS routine \texttt{?gemm} for the generic $\mA^T \mathbf{B}$ multiplication. To handle odd-sized matrices, we do not implement well-known strategies such as peeling or padding, since these are known for introducing computational and memory overhead. Instead, we manage sums between matrices of discordant size by conveniently applying the BLAS routine \texttt{?axpy} for array sums, so that it simulates padding of an extra 0 column or row, by excluding 
the last row and/or column of a sub-matrix from the sum.

\ATA and \Strassen are designed to be efficient alternatives to the BLAS routines \texttt{?gemm} and \texttt{?syrk}. Thus, they perform the same operations, respectively $\mC = \alpha \mA^T\mB + \beta\mC$ and $\mC = \alpha \mA^T\mA + \beta\mC$. However, we avoid introducing the scaling factor $\beta$ from our algorithms for clarity of exposure, since $\mC$ can be simply scaled before applying the algorithms.

\begin{algorithm}[h]
\SetAlgoLined
\KwIn{ $\mA\in \mathbb{R}^{m\times n}, \mC \in \mathbb{R}^{n \times n}, \alpha \in \mathbb{R}$}
\KwOut{ Lower triangular part of $\mC = \alpha \mA^T \cdot \mA + \mC$}
\SetKwFunction{proc}{\ATA}
\SetKwProg{myproc}{Procedure}{}{}
\myproc{\proc{$\mA$, $\mC$, $\alpha$}}{
 \label{line:ATA1}
\eIf {$m\times n \le$ cache size}
{$\mC \gets \mC +$ \texttt{blas\_?syrk}($\mA$, $\alpha$)\; 
 return\;
}
{
 Initialize pointers to $\mA_{i,j}$ and $\mC_{i,j}$, $i,j = 1,2$\;
 \ATA($\mA_{1,1}, \mC_{1,1}$, $\alpha$)\; 
 \ATA($\mA_{2,1}, \mC_{1,1}$, $\alpha$)\; 
 \ATA($\mA_{1,2}, \mC_{2,2}$, $\alpha$)\; 
 \ATA($\mA_{2,2}, \mC_{2,2}$, $\alpha$)\; \label{line:ATAlast}
 \Strassen($\mA_{1,2}$, $\mA_{1,1}$, $\mC_{2,1}$, $\alpha$)\; \label{line:strassen1}
 \Strassen($\mA_{2,2}$, $\mA_{2,1}$, $\mC_{2,1}$, $\alpha$)\; \label{line:strassen2}
}
}
\hrulefill\\
\KwIn{$\mA\in \mathbb{R}^{m\times n}, \mB\in\mathbb{R}^{m\times k}, \mC \in \mathbb{R}^{n \times k}, \alpha \in \mathbb{R}$}
\KwOut{$\mC = \alpha \mA^T \cdot \mB + \mC$}
\SetKwFunction{proc}{\textsc{FastStrassen}}
\SetKwProg{myproc}{Procedure}{}{}
\myproc{\proc{$\mA$, $\mB$, $\mC$, $\alpha$}}{
Allocate $\mM = \mathbf{0}^{n\times k / 2}$\;
Allocate $\mathbf{P} = \mathbf{0}^{m\times n / 2}$\;
Allocate $\mathbf{Q} = \mathbf{0}^{m\times k / 2}$\;
\textsc{Strassen}($\mM$, $\mathbf{P}$, $\mathbf{Q}$, $\mA$, $\mB$, $\mC$, $\alpha$)\;

 }
\caption{\ATA - Serial}
\label{alg: AtA}
\end{algorithm}

\subsection{Computational Complexity}
\label{sec:ATAT(N)}
The idea behind Strassen's algorithm is to perform a $2\times 2$ matrix multiplication using 7 multiplications instead of 8, as required by naive matrix multiplication \cite{strassen1969gaussian}. Nevertheless, Strassen's algorithm involves 18 sums between sub-matrices, thus leading to a computational complexity $T_S(n)\approx
    7 n^{\log_{2} 7}$.

In Algorithm \ref{alg: AtA}, there are four recursive calls to \textsc{AtA} on basically halved dimensions, two calls to \Strassen and 3 sums. Thus, we can derive the recurrence function for \ATA runtime depending on the input size $n$ as follows: 
\begin{equation}
\label{eq:ATACC}
    T(n)
    =
    4 T\left( \frac{n}{2} \right)
    +
    2 T_S\left( \frac{n}{2} \right)
    +
    3 \left( \frac{n}{2} \right)^2
    \approx
    \frac{2}{3} T_S(n) .
\end{equation}
 The overall computational complexity of \ATA reduces the one of the general matrix multiplication $\mA^T\mA$, amounting to $n^2(n+1)$, and of Strassen's algorithm naively applied for computing $\mA^T\mA$, that would require the same number of products as for the general matrix multiplication, and only 16 sums instead of the 18 matrix additions in the original Strassen's formulation.

\subsection{Space complexity}
\label{sec:ATAS(n)}


In \ATA, at each recursive step, pointers to the current portions of $\mA$ and $\mC$ are initialized so that, when the condition for the base-case occurs, the matrix multiplications are carried out on the correct sub-matrices of $\mA$, and stored in the corresponding locations in $\mC$. 

Strassen's algorithm for general matrix multiplication is called twice. One drawback of the naive Strassen implementation is the great amount of memory allocated at each recursive step to store the results of the intermediate matrix additions required by the algorithm. In order to avoid frequent memory allocations and releases, we call recursive Strassen (\textsc{Strassen}) on pre-allocated matrices, $\mathbf{M}$, $\mathbf{P}$ and $\mathbf{Q}$ (\Strassen). The size of such matrices is sufficiently large to store all intermediate matrix operation results throughout the recursive calls.
In fact, given an $n \times n$ matrix, at each recursive step we halve both the dimensions, rounding up the result to the nearest integer when matrices have odd sizes. By doing so, the amount of memory used by the algorithm when the base case is reached if
\begin{equation}
    \sum_{i = 1}^{\log_2 n} \frac{(n + \log_2 n)^2}{4^i}
    =
    (n + \log_2n)^2 \left (\frac{1}{3} - \frac{4}{3n^2}\right )
    \leq
    \frac{n^2}{2}
\end{equation}
which, multiplied by the three supporting matrices $\mM$, $\mathbf{P}$ and $\mathbf{Q}$, results in a total of $\frac{3}{2}n^2$.
Although the overall space complexity of Strassen does not change, we are able to save time for memory allocation at each recursive step. 
Consequently, the space complexity of \ATA is $S(n) = \frac{3}{2}n^2$.

In Section~\ref{sec:ATAtest}, we show that Strassen's algorithm benefits from the described strategy for memory allocation.


\subsection{Cache Complexity}
In this section, we show the cache complexity of \ATA. We assume the ideal cache model and we denote with $M$ the cache size, and with $b$ the size of the cache line. 
\begin{proposition}
The cache complexity of \ATA, $C_{\ATA}(n; M, b)$, is the same as the cache complexity of Strassen, $C_S(n; M, b) = \Theta(1 + \nicefrac{n^2}{b} + \nicefrac{n^{\log_2(7)}}{b\sqrt{M}})$, \cite{frigo1999cache}. 
\end{proposition}
\begin{proof}

\begin{algorithm}[h]
\SetAlgoLined
\KwIn{$\mA\in \mathbb{R}^{m\times n}, \mB\in\mathbb{R}^{m\times k}, \mC = \mathbf{0}^{n \times k}$}
\KwOut{$\mC = \mA^T \cdot \mB$}
\SetKwFunction{proc}{\textsc{RecursiveGEMM}}
\SetKwProg{myproc}{Procedure}{}{}
\myproc{\proc{$\mA$, $\mB$, $\mC$}}{
 \label{line:ATA1}
\If {$m\times n + m\times k\leq$ cache size}
{
$\mC +=$ \texttt{blas\_?gemm}($\mA^T$, $\mB$)\; 
 return\;
}
\For{$i = 1,2$}
{
    \For{$j = 1,2$}
    {
        \For{$k = 1,2$}
        {
            \textsc{RecursiveGEMM}($\mA_{k,i}$, $\mB_{k, j}$, $\mC_{i,j}$)\;
        }
    }
}

}
\caption{\textsc{RecursiveGEMM}}
\label{alg: AtAmod}
\end{algorithm}

We prove the thesis by induction. First, we observe that $C_{\ATA}(2; M, b) = 6C_S(1; M, b) \leq 7C_S(1; M, b) = C_S(2; M, b)$. Assuming as inductive hypothesis that $C_{\ATA}(\nicefrac{n}{2}; M, b) \leq C_S(\nicefrac{n}{2}; M, b)$, it holds that: \begin{align*}
    C_{\ATA}(n; M, b) &= 4C_{\ATA}(\nicefrac{n}{2}; M, b) + 2C_S(\nicefrac{n}{2}; M, b)\\
    &\leq 6 C_{S} (\nicefrac{n}{2}; M, b)\leq 7C_S(\nicefrac{n}{2}; M, b) = C_S(n; M, b).
\end{align*}
Furthermore, notice that: $C_S(\nicefrac{n}{2};  M, b) \leq C_{\ATA}(n; M, b) \leq C_S(n; M, b)$. Hence, the thesis holds. 
\end{proof}

\section{Parallel \ATA}
Our algorithm for the $\mA^T\mA$ product, \ATA, can be conveniently parallelized to work on both shared and distributed-memory systems. We will refer to our shared and distributed-memory algorithms for $\mA^T\mA$ as \ATAS and \ATAP, respectively.  Our parallel implementations of \ATA take advantage of the recursive nature of \ATA to distribute tasks (and possibly data) to different processes in an efficient way. To do so, an initial phase that implements a scheduler covering the recursion tree of \ATA is integrated in both parallel algorithms. In this way, we assign a task to each different parallel process, as we explain in Section~\ref{sec:simulexec}. After this preliminary phase, each process knows which sub-problem it has to solve.

\subsection{Preliminary phase: task assignment}
\label{sec:simulexec}
Usually, recursive algorithms are parallelized with a fork-join paradigm, according to their natural behaviour: at each recursive call, a new thread is created to accomplish that call. However, repeatedly creating and killing threads introduces a significant overhead, especially when it happens as a nested procedure. A parallelized \texttt{for} loop approach can usually improve this thread start-up overhead. For this reason, rather than addressing the problem by distributing recursive calls between newly created threads, we simulate the behaviour of a fork-join algorithm to
determine, for each thread, on which sub-matrices it must work. This is particularly useful to generalize our approach to both shared memory and distributed settings.

\begin{figure*}
    \centering
    \includegraphics[height=.38\paperwidth]{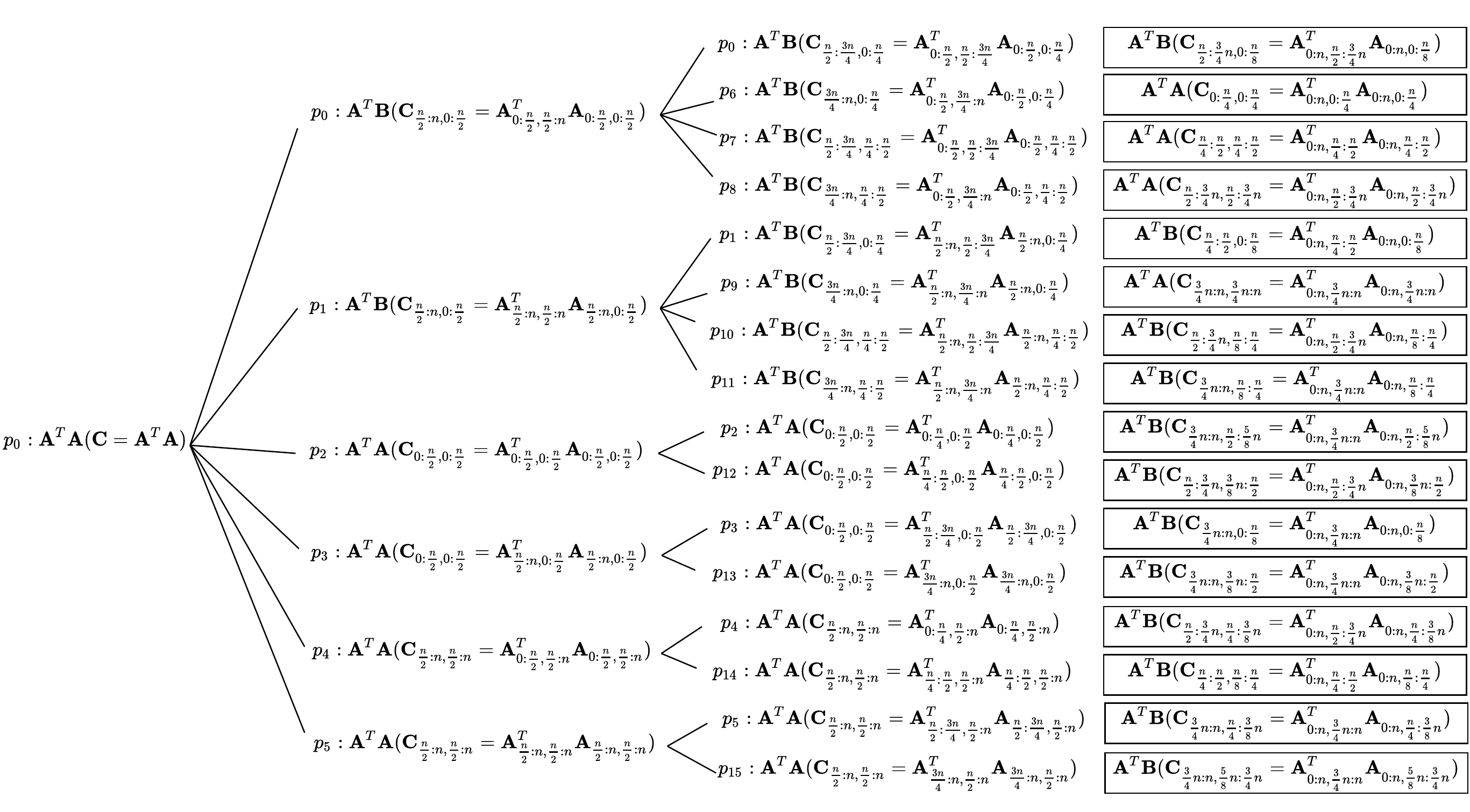}
    \caption{A tree of 16 processes distributing 
    $A \in \mathbb{R}^{n\times n}$. Boxed labels on the right-hand side are the leaf nodes of the tree generated by \ATAS, corresponding to computation tasks assigned to corresponding processes in the left-hand side leaf labels. }
    \label{fig:exectree}
\end{figure*}
\subsubsection{Building the task tree}
\label{sec:taskassignment}
To conveniently distribute tasks among $P$ parallel processes collaborating to compute $\mA^T\mA$, in the first phase of our algorithms, each process builds the recursion tree of a modified version of \ATA, that we shall call \ATAnaive, and explores a part of it with a breadth-first search (BFS), see Figure~\ref{fig:exectree}. \ATAnaive considers classic recursive general matrix multiplication instead of Strassen, and can be easily implemented by modifying 
Algorithm~\ref{alg: AtA} to call \textsc{RecursiveGEMM} instead of \textsc{Strassen}. \textsc{RecursiveGEMM}, summarized in Algorithm~\ref{alg: AtAmod}, is a recursive algorithm for the naive general matrix multiplication. 
The reasons behind this choice will be explained in Section~\ref{sec:naiveVSstrassen}. We define the \emph{task tree}, denoted with $\TaskTree$, to be the sub-tree of the recursion tree of \ATA, obtained by spanning the latter with a BFS, that is interrupted as soon as $\TaskTree$ counts $P$ leaves, labeled from $0$ to $P - 1$.
Both \ATAS and \ATAP implement the task tree, but with some 
differences concerning data and task division. 
In \ATAP, each $p$-th leaf corresponds to the task that process $p$ has to fulfil, and contains directives on both the computational and communication activity that is due to the corresponding process. Specifically, a leaf task $t$ provides the following information:
\begin{enumerate}
    \item $t.computationType$: Which type of computation process $p$ has to carry out. It can be either a $\mA^T\mA$ or a $\mA^T\mB$ multiplication;
    \item $t.\mathbf{X}.offset$ and $t.\mathbf{X}.q$, with $\mathbf{X}\in\{\mA, \mB, \mC\}$, $q\in\{m,n\}$: The row and column offsets as well as the size of the sub-matrices of $\mA$ and $\mC$ process $p$ has to work on;
    \item $t.parent$: The parent process that sends sub-matrices of $\mA$ to its children (during the distribution phase), and to which process $p$ has to send the result of the task that was assigned to it or, if $p$ is the parent, the information on its children' tasks (during result retrieval). 
\end{enumerate}
Inner nodes of $\TaskTree$ instead, represent tasks concerning data distribution and retrieval, possibly involving sums of sub-matrices of $\mC=\mA^T\mA$, and consequent communication (point 3 of the previous list), and are executed by a subset of processes. In contrast, in \ATAS only leaf nodes of $\TaskTree$ correspond to a task, whereas inner nodes are ignored, as no communication is involved. For the same reason, leaf tasks only include information about what kind of computation the corresponding threads have to carry out and on the sub-matrices they have to work on (points 1 and 2 of the previous list).
\subsubsection{Load Balancing}
The task tree of \ATAP is created so that, at each level, given $P$ available processes, $\alpha\cdot P$ processes compute a general $\mA^T\mB$ matrix multiplication; for the remaining $(1-\alpha)\cdot P$ processes, a task for a $\mA^T\mA$ multiplication is assigned to them. Here, $\alpha\in (0,1)$ is a parameter for balancing the workload among distributed processes, as the computational complexity of a $\mA^T\mA$ product is lower than the one of $\mA^T\mB$. 
The task tree $\TaskTree$ is built by calling \textsc{RecursiveGEMM} (whose computational complexity is roughly twice the one of \ATA, $T(n)$). Therefore the number of multiplications carried out in $\TaskTree$ to perform $\mA^T\mB$ is twice the one needed to compute $\mA^T\mA$. The load balancing parameter must be such that $4\cdot \nicefrac{T(n)}{(1-\alpha)P}=2\cdot \nicefrac{2T(n)}{\alpha P}$. In accordance, we set $\alpha = \nicefrac{1}{2}$. 
 This task division is repeated recursively at each level, by progressively decreasing the number of available processes, $P$. 
The number of recursive parallel steps depends on $P$ and $\alpha$. In particular, for $\alpha = 0.5$, the number of parallel levels in the task tree, $\ell$ is given by the following expression:
\begin{equation}
\label{eq:numLevels}
\begin{aligned}
&\ell(P =1) = 0, \quad \ell(2\leq P \leq 6) =1\\
&\ell(P>6) =  1 + k + \text{sign} \left ({\frac{P}{4}}\mod {8^{\max\{k; 1\}}}\right ),
\end{aligned}
\end{equation}
where $k = \max\left \{k \in \mathbb{N}: \frac{\nicefrac{P}{4}}{8^k} \ge 1\right \}$ and $\mathrm{sign}(x)$ is the sign function, returning $0$ for $x = 0$ and $1$ for $x > 0$.
Indeed, when \ATAP is called on $P$ processes, $\nicefrac{P}{2}$ of them are going to compute $\mC_{2,1}$; out of them, $\nicefrac{P}{4}$ processes compute $\mA_{1,2}^T\mA_{1,1}$, whereas the remaining $\nicefrac{P}{4}$ are in charge for $\mA_{2,2}^T\mA_{2,1}$ (see Equation~\ref{eq: Cij}). These tasks are in turn distributed among 8 processes each, recursively (corresponding to the eight recursive calls of \textsc{RecursiveGEMM}). This splitting is repeated as long as it possible (i.e., until $\nicefrac{P}{4}/8^k\ge 1$). If by doing so, all $\nicefrac{P}{4}$ processes are used (i.e., $\nicefrac{P}{4}$ is a multiple of $8^k$, for some $k$), all processes work on equally sized matrices. Otherwise, some processes will further split their tasks to smaller matrices, resulting in an additional parallel level. We say that the last parallel level is \emph{complete} when all leaves corresponding to $\mA^T\mA$ tasks are grouped in bunches of 6 siblings, and when all leaves corresponding to $\mA^T\mB$ tasks are grouped in bunches of 8 siblings.

The task tree for \ATAS is quite different. In order to avoid concurrent overlapping writes, input matrices are tiled in horizontal and vertical blocks, as depicted in Figure~\ref{fig:matmul-verthor-tiling}. 
 This way, we ensure that each thread computes a different $\mC_{i, j}$. With this new scheme, we make three recursive calls to \ATA (instead of 6) and four recursive calls to \Strassen (instead of 8). 
Therefore, the number of parallel levels in \ATAS, given $P$ threads, is the following:
\begin{equation}
    \begin{aligned}
        &\ell(P=1) = 0,\quad \ell(P=2,3) = 1,\\
        &\ell(P>3) = 1 + k + \mathrm{sign}\left( {\frac{P}{2}}\mod {4^{\max \{k; 1\}}} \right),
    \end{aligned}
\end{equation}

with $k = \max\left \{k \in \mathbb{N}: \frac{\nicefrac{P}{2}}{4^k} \ge 1\right \}$.
In Figure~\ref{fig:exectree}, we show an example of the task tree with 16 processes for \ATAP, and the leaf nodes of the task tree for \ATAS (boxed). 

\subsubsection{Naive matrix multiplication over Strassen}
\label{sec:naiveVSstrassen}
In our parallel algorithms, we do not rely on Strassen for general $\mA^T\mB$ matrix multiplication when building the recursion tree, that instead is created by simulating \ATAnaive. This is done with the goal of optimizing the resources of distributed architectures, as the naive general matrix-multiplication algorithm does not allocate the additional memory required by Strassen, resulting in a faster memory management. Furthermore, Strassen's algorithm would possibly cause a hardly manageable workload unbalance between processes implementing an $\mA^T\mA$ multiplication, and those that would be in charge of computing the intermediate matrix sums appearing in Strassen's algorithm. However, Strassen's algorithm can still be used at leaf-level computation.



\subsection{Shared-memory \ATA}
\label{sec:ATAS}
\ATA can be implemented with a shared-memory parallel paradigm on multi-core machines.
We rely on OpenMP to efficiently distribute the workload between threads. Each thread simulates the recursion of \ATAnaive as described in Section~\ref{sec:simulexec}. The workload is distributed so that each thread writes in a different memory location, hence there is no need of handling data collisions of any kind. Instead, the problem is divided in a fashion that makes it embarrassingly parallel. We call \ATAS our multi-threaded algorithm for $\mA^T\mA$.

\subsubsection{\ATAS in detail}
Let us denote with $P$ the number of available threads. 
Our algorithm for multi-threaded machines, \ATAS, can be divided into two phases. During the first phase, one task is assigned to each thread by simulating the recursion of \ATAnaive, as described in Section~\ref{sec:simulexec}.
In order to prevent memory collisions and to achieve embarrassing parallelism, tasks are organized so that each thread writes on a different and disjoint memory location. This is done by dividing the resulting matrix $\mC$ into four blocks, as shown in Equation~\ref{eq: Cij}, whereas $\mA$ is tiled vertically or horizontally, instead of in $2 \times 2$ blocks
(see Figure~\ref{fig:matmul-verthor-tiling}). 
This procedure avoids concurrent writing management, it guarantees data and thread reuse and relies on the equality:
\begin{equation}
\label{eq:tiledmult}
    \mC_{i, j}
    =
    \mA_{i, 1}\mB_{1, j} + \mA_{i, 2}\mB_{2, j}
    =
    \mA_{i, *}\mB_{*, j},
\end{equation}
for $i,j=1,2$. Such instruction and data assignment allows for a faster execution, since threads never need to synchronize. 

During the second phase of \ATAS, each thread retrieves its task from the tree $\TaskTree$, specifying which routine (either \ATA or \Strassen) the corresponding thread must call, and on which sub-matrices of $\mA$ and $\mC$ it must operate.  On multicore systems, this means that data reuse in both L1 and L2 cache is optimized, since each thread operates on the same data throughout its entire lifespan. 
Since the tasks correspond to disjoint sub-problems, 
at the end of the computation each thread only needs to synchronize with the others, then the algorithm stops. In Algorithm~\ref{alg: AtAS} we provide the pseudo-code of \ATAS.

\begin{algorithm}[h]
\SetAlgoLined
\KwIn{ $\mA\in \mathbb{R}^{m\times n}$ }
\KwOut{ Lower triangular part of $\mC = \mA^T \cdot \mA$}
\SetKwFunction{proc}{\textsc{\ATAS}}
\SetKwProg{myproc}{Procedure}{}{}
\myproc{\proc{$\mA$}}{
Generate tree $\TaskTree$\;
\textbf{par}\For{each leaf-node $v$ of $\TaskTree$}
{
    Get task $t$ from node $v$\;
    \uIf{$t.computationType = \mA^T\mA$}
        {\ATA ($\mA_{t.\mA.offset}$, $\mC_{t.\mC.offset}$, 1)\;} 
    \uElseIf{$t.computationType = \mA^T\mB$}{ 
        \Strassen~($\mA_{t.\mA.offset}$,~$\mA_{t.\mB.offset}$,~$\mC_{t.\mC.offset}$,~$1$)\;
        \vspace{-0.4cm}
    }
}   
}
\caption{\ATAS - Shared}
\label{alg: AtAS}
\end{algorithm}

\subsubsection{Computational Complexity of \ATAS}
\label{sec:ATAStime}
We study the time complexity $T(n, P)$ of \ATAS to perform the multiplication $\mA^T\mA$ on an $n \times n$ matrix $\mA$ and distributing the workload between $P$ processes.

At first, the algorithm needs to generate the task tree and each process has to retrieve its task. These procedures have the same complexity as a BFS visit on a tree with $P$ leaves, hence $O(P)$.

The time complexity of the second step corresponds to the 
one of the most expensive leaf task, which appears at the end of a path of \textsc{RecursiveGEMM} calls.
At level $l$, the size of the product matrix $\mC$ is reduced to a block of size $\nicefrac{n}{2^l}\times\nicefrac{n}{2^l}$, resulting from a multiplication between  $\nicefrac{n}{2^l} \times n$ and $n\times\nicefrac{n}{2^l}$ matrices. Thus, the total complexity is reduced by $4^{\ell(P)}$, being $\ell(P)$ the number of levels in the task tree. Hence the total complexity of the algorithm is:
\begin{equation}
    T(n, P)
    =
    O(P)
    +
    O\left(
    \frac{1}{4^{\ell(P)}}
    n^{\log_2 7}
    \right).
\end{equation}
Notice that $\ell(P)$ is a discrete, non-injective function. Hence, especially with few processes, the
speed-up behaves like a step function. Despite this behaviour, $\ell(P) \approx \log_4 P$, meaning with large numbers of processes we achieve a theoretical linear speed-up.

\begin{figure}[h!]
    \centering
    \includegraphics[width=.9\columnwidth]{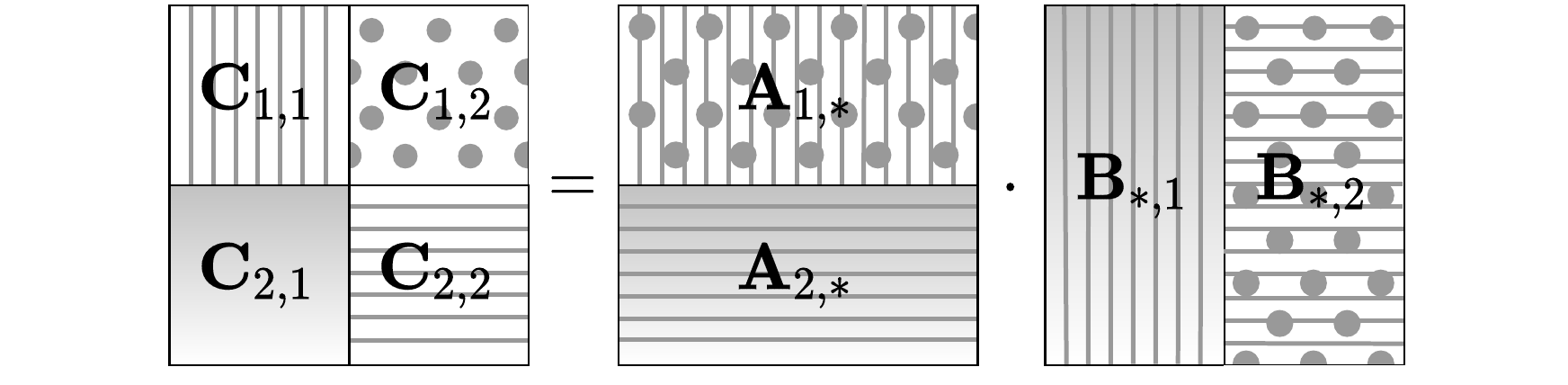}
    \caption{Multiplication with vertical/horizontal tiling.}
    \label{fig:matmul-verthor-tiling}
\end{figure}

\subsection{Distributed-memory \ATA}
\label{sec:ATAP}
Modern computers are equipped with an ever-increasing number of cores inside CPU chips. However, when it comes to massive volumes of data, computationally intensive tasks such as matrix multiplication are simply prohibitive, even for the most recent 16- or 32-cores chipsets, and even with hyper-threading capabilities. Distributed parallelism plays a crucial role in this setting, as it allows to distribute the workload between multiple machines. In such an environment, providing fast distributed algorithms for operations in Linear Algebra, including $\mA^T\mA$ multiplication, is a key task to limit bottlenecks.


In this section, we describe a distributed algorithm for $\mA^T\mA$, that works for any matrix size and with arbitrarily many processes and cores. We shall refer to this algorithm as \ATAP. \ATAP follows a distribute-compute-retrieve paradigm, as initially the input matrix $\mA$ is stored on the root process only, and distributed to other processes according to their tasks. Finally, the resulting matrix $\mC=\mA^T\mA$ is retrieved back by the root process. We implement a parallel communication scheme to limit data transfer overhead.

\begin{algorithm}[h]
\SetAlgoLined
\KwIn{$\mA\in \mathbb{R}^{m\times n}$}
\KwOut{Lower triangular part of $\mC = \mA^T \cdot \mA$}
\SetKwFunction{proc}{\textsc{\ATAP}}
\SetKwProg{myproc}{Procedure}{}{}
\myproc{\proc{$\mA$}}{
Generate tree $\TaskTree$\;
\For{each $v$ of $\TaskTree$ in the path from my leaf to the root}
{   
    Get my task $t$ from node $v$\;
    \If{$v$ is a leaf}
    {
        \uIf{$t.computationType = \mA^T\mA$}
        {
        $\mC_{t.\mC.offset}$ = $\mA^T\mA$($\mA_{t.\mA.offset}$)\; \label{line:baseCaseATA} 
        }
        \uElseIf{$t.computationType = \mA^T\mB$}
        {$\mC_{t.\mC.offset}$ = $\mA^T\mB$($\mA_{t.\mA.offset}$, $\mA_{t.\mB.offset}$)\;\label{line:baseCaseATB}
        }
    }
    \eIf{$t.parent \neq$ my ID}
    {
    Send $\mC_{t.\mC.offset}$ to $t.parent$\;\label{line:parentsend}
    }
    {
    Receive $\mC_{children.t.\mC.offset}$ from my children\;
     Sum over the sub-matrices and store result in $\mC$\;
    }
}   
}
\caption{\ATAP - Distributed}
\label{alg: AtAP}
\end{algorithm}

\subsubsection{\ATAP in detail}
\label{sec:ATAPdet}

Let $P$ be the number of distributed processes. In \ATAP, each process $p$ first builds the task tree $\TaskTree$ as described in Section~\ref{sec:simulexec}. To understand in detail how $\TaskTree$ is used in \ATAP, we shall refer to the example of Figure~\ref{fig:exectree}. As we said, each node represents a task, but only tasks contained in leaf nodes correspond to an actual matrix multiplication. Inner nodes instead represent tasks assigned only to the parents of the nodes branching out of them, and they are necessary to retrieve and combine the portions of the result matrix scattered among different processes, and eventually to send them, level by level, up to the root process, $p_0$. In the example of Figure~\ref{fig:exectree}, $\TaskTree$ is the task tree for $P=16$ processes on a square matrix. Leaf nodes are generated so that processes $p_0$, $p_1$ and $p_6\ldots,p_{11}$ share the workload to compute $\mC_{2,1}$. The remaining half of the processes is devoted to compute $\mC_{1,1}$ and $\mC_{2,2}$. If the number of distributed processes is not enough to make a complete level, as in this example, instead of calling multiple tasks on different tiles of the matrices, processes perform either an $\mA^T\mA$ or a $\mA^T\mB$ operation on vertically and horizontally tiled sub-matrices at the leaf-level. For instance, observe the first batch of sibling-leaves in Figure~\ref{fig:exectree}. To compute $\mathbf{C}_{\nicefrac{n}{2}:n, 0:\nicefrac{n}{2}}= \mathbf{A}_{0:\nicefrac{n}{2}, \nicefrac{n}{2}:n}^T \mathbf{A}_{0:\nicefrac{n}{2}, 0:\nicefrac{n}{2}}$, \ATAnaive would perform 8 recursive calls to $\mA^T\mB$; in \ATAP, each of these calls is served by one distributed process, if available. When this is not the case, as in the example that we are considering, processes $p_0, p_6,p_7,p_8$ divide $\mathbf{A}_{0:\nicefrac{n}{2}, \nicefrac{n}{2}:n}$ and $\mathbf{A}_{0:\nicefrac{n}{2}, 0:\nicefrac{n}{2}}$ in  vertical tiles so as to compute the related portions of $\mC$ as depicted in Figure~\ref{fig:matmul-verthor-tiling}. When the computation is over, partial results are collected by the parents of each group of siblings (processes $p_i$, $i=0,\ldots,5$). 
This operation is iterated by traversing the tree up to its root, $p_0$, and allows for a convenient parallel communication reducing data transfer overhead. In order to optimize the communication and to reduce the exchanged data volume, we encode the sub-matrices resulting from $\mA^T\mA$ operations as packed lower triangular matrices. Nevertheless, the entire operation, once it returns to the root process, still produces a standard square matrix. 
In Algorithm~\ref{alg: AtAP}, we provide the pseudocode of \ATAP. In line~\ref{line:parentsend}, if the process has to fulfill a $\mA^T\mB$ task, it sends to its parent the entire sub-matrix $\mC_{t.\mC.offset}$; otherwise, it only sends $\text{low}(\mC_{t.\mC.offset})$. In lines~\ref{line:baseCaseATA} and \ref{line:baseCaseATB}, $\mA^T\mA$ and $\mA^T\mB$ may refer to \ATA or \texttt{blas\_?syrk}, and to \Strassen or \texttt{blas\_?gemm}, respectively. As we shall see in Section~\ref{sec:ATAtest}, the real benefit of using our implementation of \ATA and \Strassen arises on matrices with larger size, therefore they are favourable when handling larger volumes of data.

\subsubsection{Computational and Communication Complexity of \ATAP}
\label{sec:CCatap}
In contrast to parallel algorithms for distributed matrices, \ATAP does not include any communication between processes at computation time, as the input matrix is scattered among distributed processes so that they own the exact portions of $\mA$ on which they have to operate.

\begin{proposition}
The computational cost of \ATAP (Algorithm~\ref{alg: AtAP}) on a matrix of size $n$ 
and with using $P$ processes, $C(n,P)$ is:
\begin{equation*}
    C(n,P) = O\left (\left ( \nicefrac{n}{2^{\ell(P)}} \right )^2\cdot \nicefrac{n}{2^{\ell(P) - 1} }\right),
\end{equation*}
if the load balancing parameter $\alpha$ is set to $0.5$.
\end{proposition}
\begin{proof}
$C(n,P)$ depends on the number of recursive levels that can be layered with the available resources and on $\alpha$. For $\alpha = 0.5$, the computational complexity of \ATAP is given by the time for computing $\mA^T\mB$ on matrices of size at most $\nicefrac{n}{2^{\ell(P)}}\times \nicefrac{n}{2^{\ell(P) - 1}}$, that is $O\left (\left ( \nicefrac{n}{2^{\ell(P)}} \right )^2\cdot \nicefrac{n}{2^{\ell(P) - 1} }\right)$, where $\ell(P)$ is the number of parallel levels defined in Equation~\ref{eq:numLevels}. 
\end{proof}

We express the communication cost for matrix distribution and result retrieval in terms of latency and bandwidth costs of a distributed algorithm, denoted with $L(n,P)$ and $BW(n,P)$, respectively, using the same definitions introduced in \cite{ballard2013graph} and adopted also in \cite{COSMA}. Latency cost is the communicated-message count, whereas bandwidth is expressed in terms of communicated-word count. Messages and words counts are computed along the critical path of the distributed algorithm, as defined in~\cite{yang1988critical}. 
\begin{proposition}
The latency of \ATAP on a matrix of size $n\times n$ and with $P$ processes is $L(n,P) = O(2[7\cdot (\ell(P) - 1) + 5])$. Its bandwidth is $BW(n,P) \leq 6(\nicefrac{n}{2})^2 + \frac{n(n+2)}{2} + \nicefrac{7}{6}n^2(1-\nicefrac{1}{4^{\ell(P)-2}})$. 
\end{proposition}
\begin{proof}
In \ATAP, the critical path corresponds to the sequence of communication operations carried out by the root process $p_0$. After the first parallel level, $p_0$ works on a $\mA^T\mB$ task and shares its workload with 7 other processes at each parallel level. When the compute phase is over, at each level $l\in \{2,\ldots,\ell(P)\}$ process $p_0$ collects partial results from its (at most) seven children; at level $l=1$, it retrieves the entire matrix $\mC=\mA^T\mA$ by combining together the results of its five siblings. This operation is carried out both for data distribution and result collection. Hence, $L(n,P) = O(2[7\cdot (\ell(P) - 1) + 5])$.\\
During the data distribution phase, message sizes (i.e., portions of input matrix $\mA$) decrease when descending from the root down to the leaves of $\TaskTree$. In the first level, $p_0$ distributes two matrices of size $\nicefrac{n}{2}\times \nicefrac{n}{2}$ to the other process that is in charge to carry out $\mA^T\mB$ tasks, and one sub-matrix of the same size to each of its four siblings that have to compute $\mA^T\mA$. For each level $l\in\{2,\ldots,\ell(P)\}$, the root process sends matrices of size $\nicefrac{n}{2^l}$ to at most 7 other processes. Hence, during the distribution phase, $BW(n,P)$ is $O(5\left (\nicefrac{n}{2}\right )^2 + 7\cdot \sum_{l = 2}^{\ell(P)}\left (\nicefrac{n}{2^l}\right )^2)=O(5(\nicefrac{n}{2})^2 + \nicefrac{7}{12}n^2 (1 - \nicefrac{1}{4}^{\ell(P)-2}))$. 
With similar considerations and taking into account the fact that processes sending symmetric portions of $\mC$ only store its lower triangular part ($\text{low}(\mC)$), it holds that the bandwidth during the result retrieval phase 
amounts to $O(\left (\nicefrac{n}{2}\right )^2 + 4 (\nicefrac{n(n+2)}{8}) + 7\cdot \sum_{l=2}^{\ell(P)} \left (\nicefrac{n}{2^{l}}\right )^2) = O(\left (\nicefrac{n}{2}\right )^2 + \nicefrac{n(n+2)}{2} +\nicefrac{7}{12}n^2 (1 - \nicefrac{1}{4}^{\ell(P)-2})$. The thesis follows by summing together the two components. 
\end{proof}

\begin{figure}[h!]
    \centering
    \subfigure[Elapsed time.\label{fig:timeATAseq}]{\includegraphics[width=0.49\columnwidth]{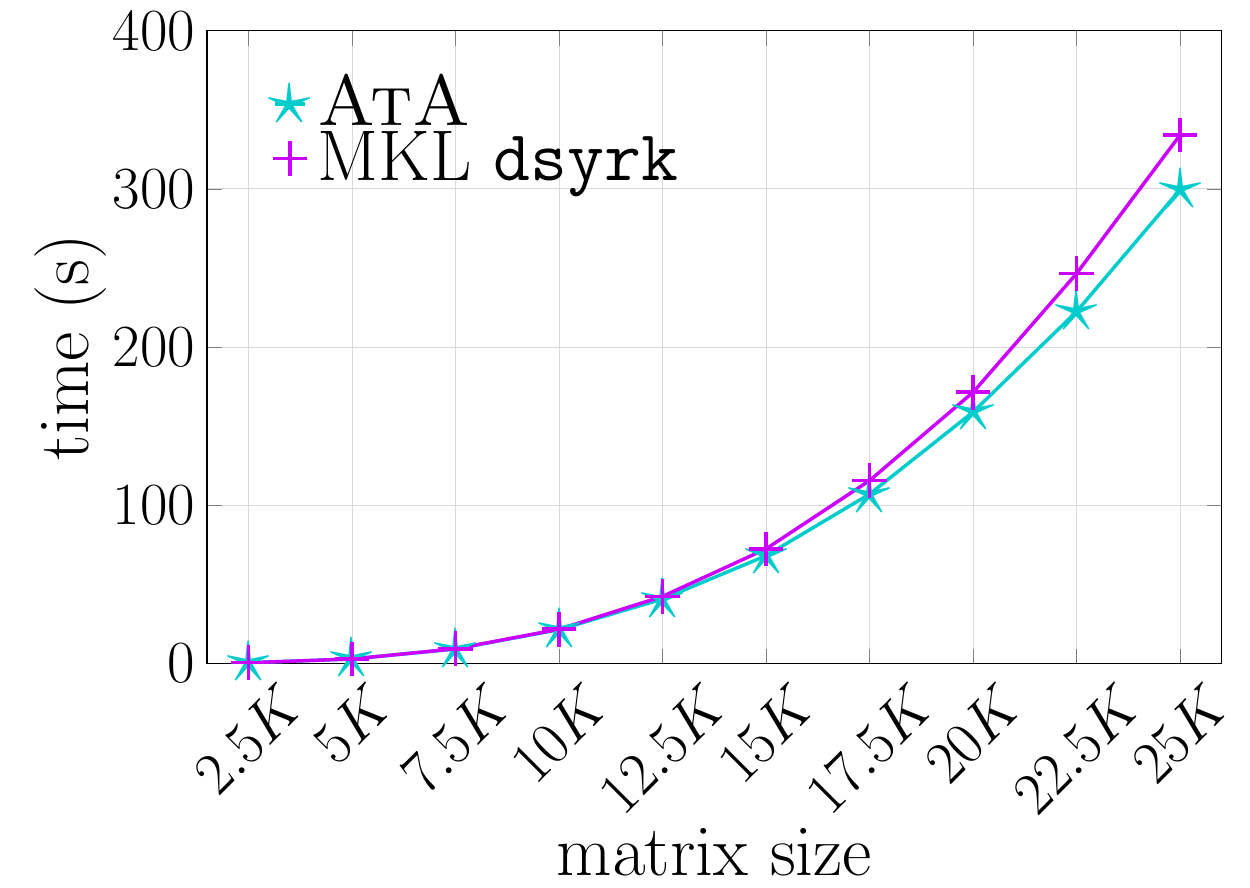}}
    \subfigure[Effective GFLOPs.\label{fig:EGATAseq}]{\includegraphics[width=0.49\columnwidth]{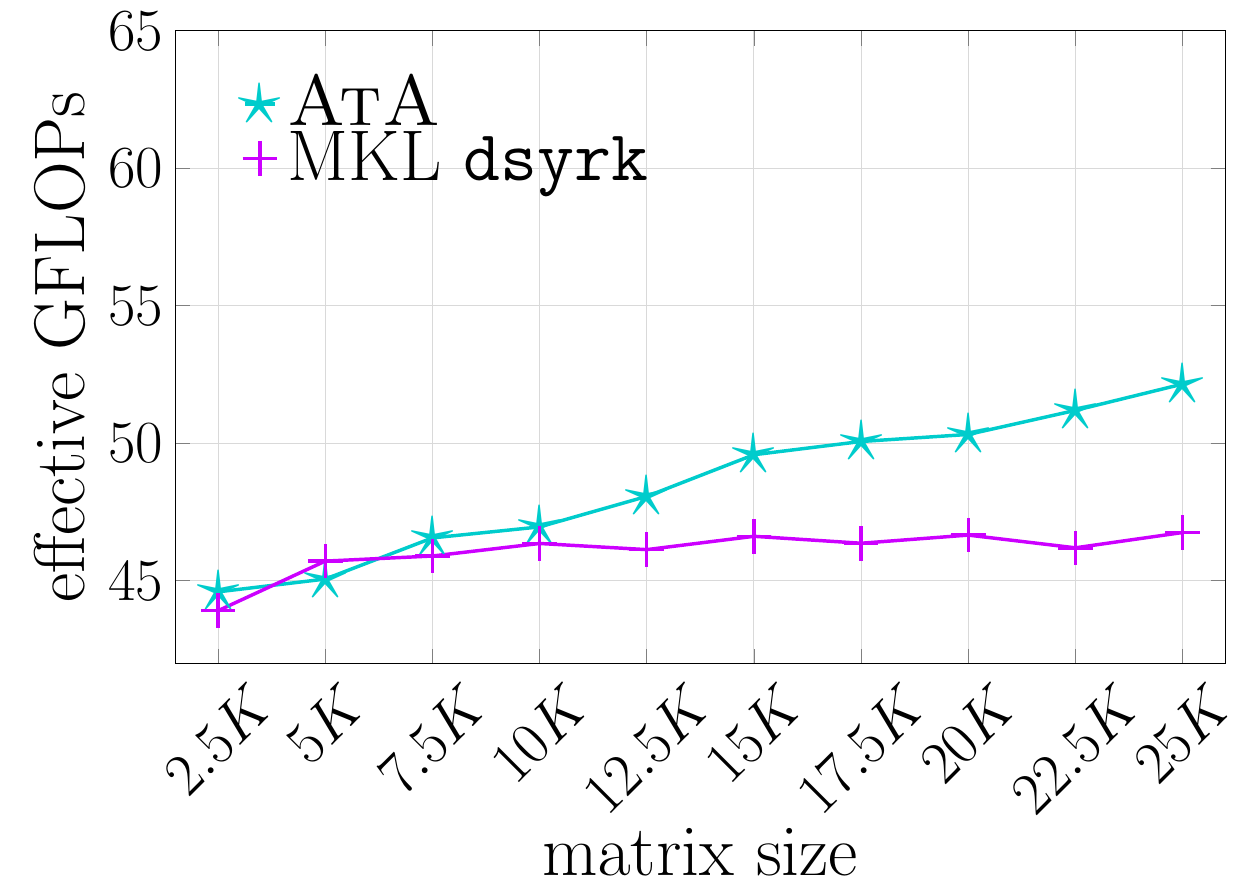}}
    \caption{\ATA vs Intel MKL \texttt{dsyrk}}
    \label{fig:ATASeq}
\end{figure}
\begin{figure}[h!]
    \vspace{-0.3cm}
    \subfigure[Elapsed time.\label{fig:timeStrassenseq}]{\includegraphics[width=0.49\columnwidth]{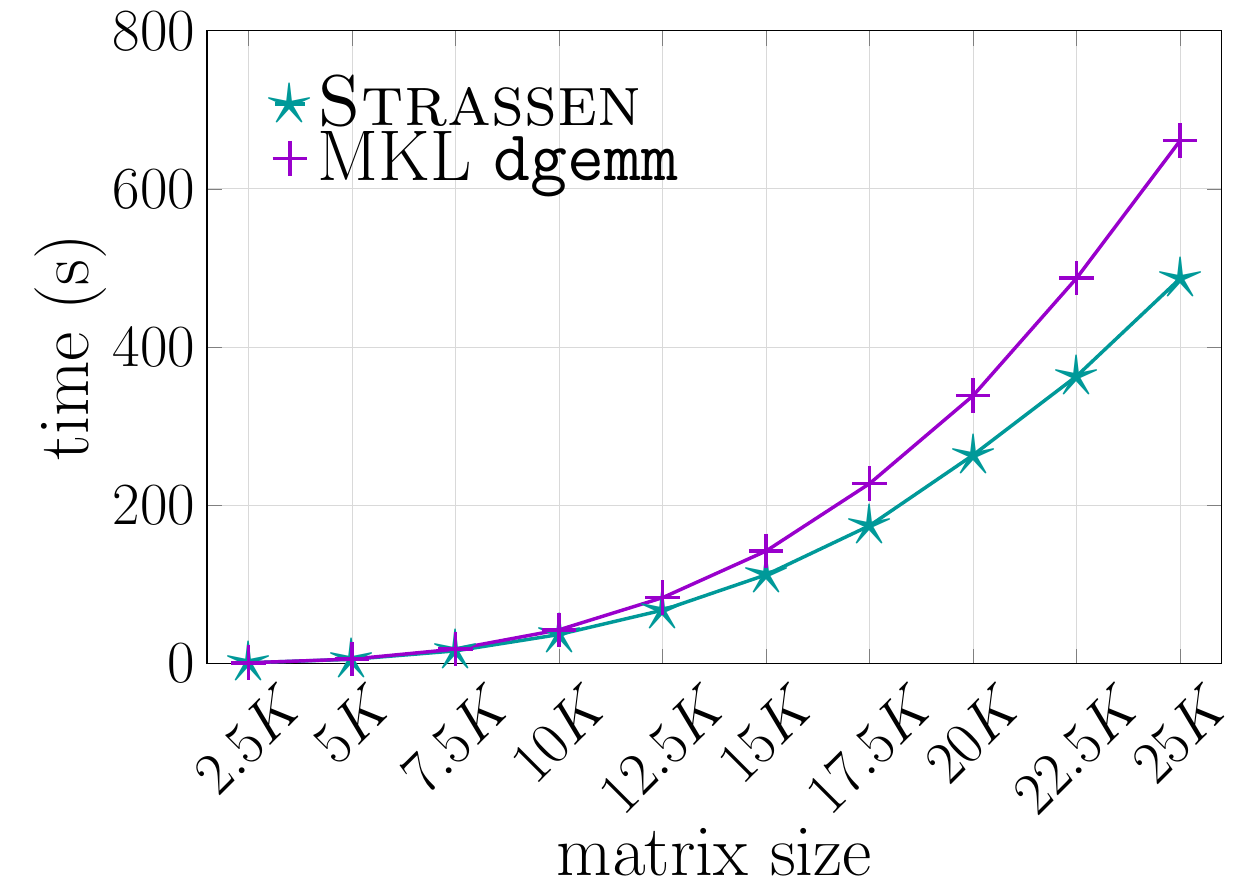}}
    \subfigure[Effective GFLOPs.\label{fig:EGtimeStrassenseq}]{\includegraphics[width=0.49\columnwidth]{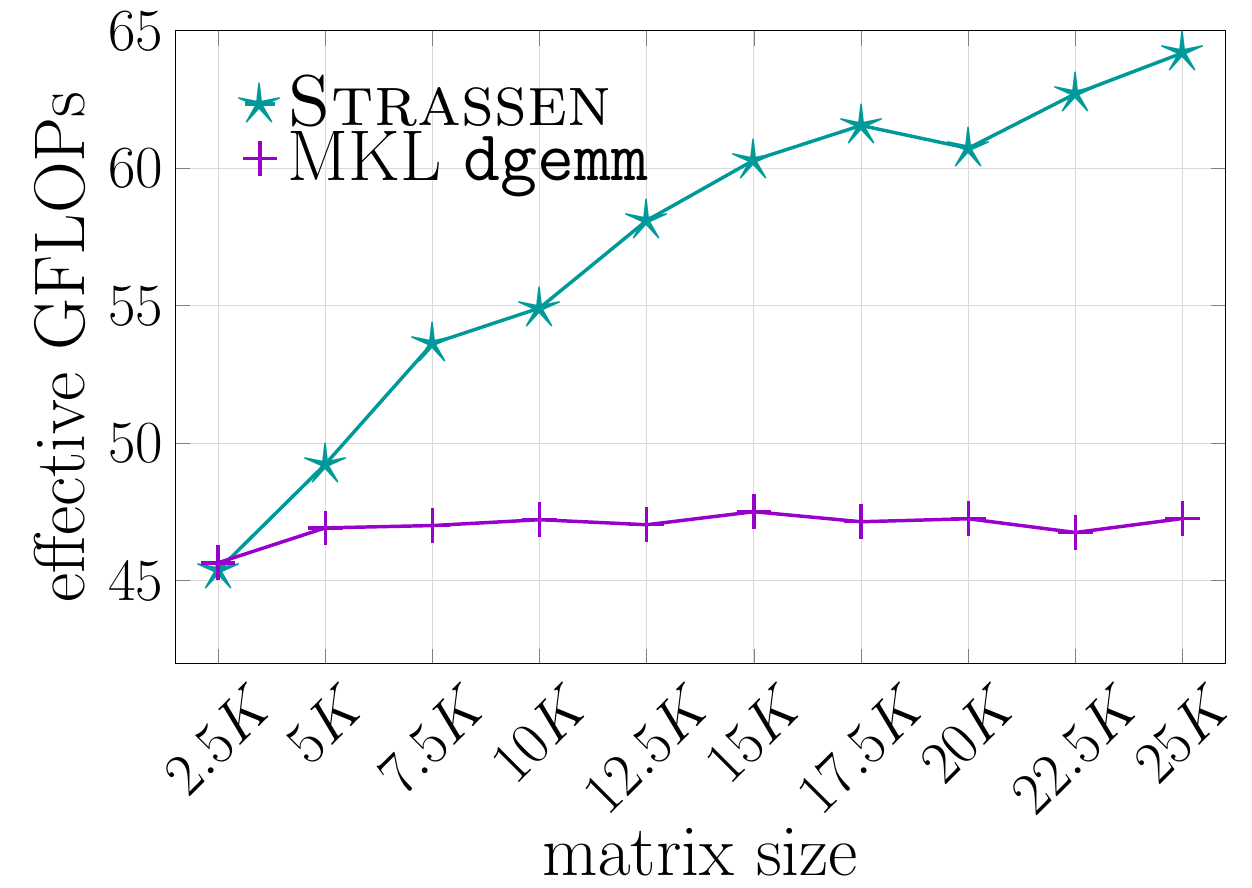}}
    \caption{\Strassen vs Intel MKL \texttt{dgemm}}
    \label{fig:StrassenSeq}
\end{figure}

From this analysis, we see that computation has the prominent role in time complexity $T(n,P)=C(n,P)+L(n,P)+BW(n,P)$. 
This fact will be confirmed by our experimental results, presented in Section~\ref{sec:ATAtest}, where we see how increasing the matrix sizes provides an always increasing benefit in using the distributed algorithm, proving that communication cost $L(n, P) + BW(n, P)$ is absorbed by the computational cost, $C(n, P)$, for growing values of $n$. 

\section{Performance Evaluation}
\label{sec:ATAtest}
We evaluate the performance of our algorithms with an extensive set of experiments over multiple benchmarks. Our code is available at \url{https://github.com/filthynobleman/AtA}.

\subsection{Experimental Setup}
All tests reported in this section were run on TeraStat\footnote{https://www.dss.uniroma1.it/en/node/6554}, a cluster of 12 compute nodes, each equipped with 2 sockets of Intel Xeon E5-2630v3 8 cores, 2.4 Ghz, 4 GB RAM per core.

We test our algorithms and benchmark solutions on square and tall matrices, generated randomly. We carry out experiments in both single and double floating-point precision, to highlight the fact that our algorithm achieves good performance in both settings.

 In the tests, we exploit the Intel Math Kernel Library (MKL) both by integrating BLAS routines for basic matrix operations, and for the validation of the proposed algorithms through performance comparisons with shared and distributed memory parallel benchmark solutions. MKL is a  framework that includes routines and functions optimized for Intel and compatible processor-based computers, and provides C/C++ interfaces and the acceleration of libraries for Linear Algebra (including BLAS and ScaLapack) within several third-party math libraries. \cite{wang2014intel,intelMKL}.

\subsection{Metrics}
To compare the performance of our algorithms against benchmark methods, we use the average elapsed time in seconds and the effective GFLOPs. Effective GFLOPs  is a measure for comparing classical and fast matrix-multiplication algorithms. For classical algorithms, which perform $2n^3$ floating point operations, Equation~\ref{eq:effGFLOPs} gives the actual GFLOPs; for fast matrix-multiplication algorithms, it gives the performance relative to classical algorithms, but does not accurately represent the number of floating point operations performed  \cite{demmel2013communication}. For fair comparisons, we calculate the metrics as: 
\begin{equation}
\label{eq:effGFLOPs}
    \text{effective GFLOPs} = \frac{r n^3}{\text{execution time in seconds}\cdot 10^9}
\end{equation}
where $r=1$ when we test algorithms specifically built for the $\mA^T\mA$ product, whereas $r=2$ when algorithms for the general matrix multiplication are tested. 

\subsection{Sequential}
Figures~\ref{fig:ATASeq} and \ref{fig:StrassenSeq} show the execution time and effective GFLOPs of the sequential \ATA and \Strassen routines, respectively. Their performance is compared to the Intel MKL counterparts: \texttt{dsyrk} and \texttt{dgemm}. The experiments are carried out on matrices of growing matrix size (from $2.5\cdot 10^3$ to $2.5\cdot 10^4$), and run on a single Intel core. 
The time difference between our solutions and the ones implemented by Intel MKL grows with the matrix size, reflecting the lower computational cost of our approach. Figure~\ref{fig:StrassenSeq} proves how Strassen's algorithm benefits from the pre-memory-allocation strategy described in Section~\ref{sec:ATAS(n)}.

\subsection{Shared memory}
For evaluating the shared memory parallel implementation of the $\mA^T\mA$ product, \ATAS, we compare it against the Intel MKL implementation of the BLAS routine \texttt{ssyrk}, for single precision symmetric rank-$K$ update. For both methods, we always use a 16 thread setup, and we analyse the execution time and the effective GFLOPs (Equation~\ref{eq:effGFLOPs} with $r=1$) while varying the number of available cores. In light of the sequential experiments shown in Figures~\ref{fig:ATASeq} and \ref{fig:StrassenSeq}, we compare \ATAS and MKL \texttt{ssyrk} on larger matrices, where tests highlight more interesting results.  In particular, we run experiments on square matrices of size $3\cdot 10^4\times 3\cdot 10^4$, $4\cdot 10^4\times 4\cdot 10^4$ and on tall matrices of size $6\cdot 10^4\times 5\cdot 10^3$. Figure~\ref{fig:shared} summarizes our results. As anticipated by the study of the computational complexity, the execution time is reduced by $\nicefrac{1}{4}$ at each complete parallel level. 
Figures~\ref{fig:shar-30time}, \ref{fig:shar-40time} and \ref{fig:shar-605time} show how our algorithm can compete with the MKL implementation when the core availability is large, and that significantly outperforms the Intel implementation in the $P \leq 10$ cores setup. Furthermore, we show in Figures~\ref{fig:shar-30gflop}, \ref{fig:shar-40gflop} and \ref{fig:shar-605gflop} that $\ATAS$ is capable not only of accomplishing a large amount of floating point operations per second, but also that its performance growth rate is consistent with the step-wise behaviour of the time complexity studied in Section~\ref{sec:ATAStime}. This justifies sporadic thinnings in performance gap between the two methods.  
\begin{figure}[t!]
    \subfigure[Elapsed time, $\mA\in \mathbb{R}^{30K\times 30K}$.\label{fig:shar-30time}]{\includegraphics[width=0.49\columnwidth]{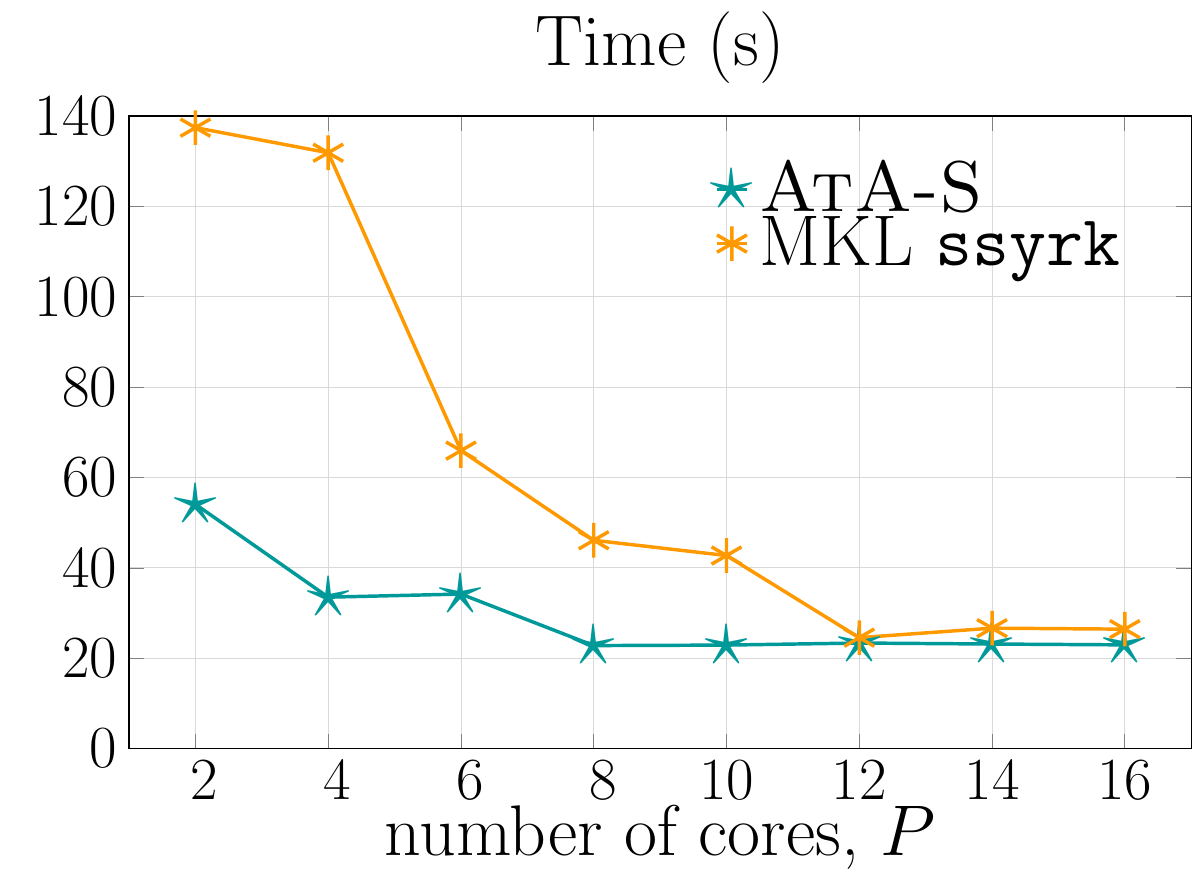}}
    \subfigure[EGs, $\mA\in \mathbb{R}^{30K\times 30K}$.\label{fig:shar-30gflop}]{\includegraphics[width=0.49\columnwidth]{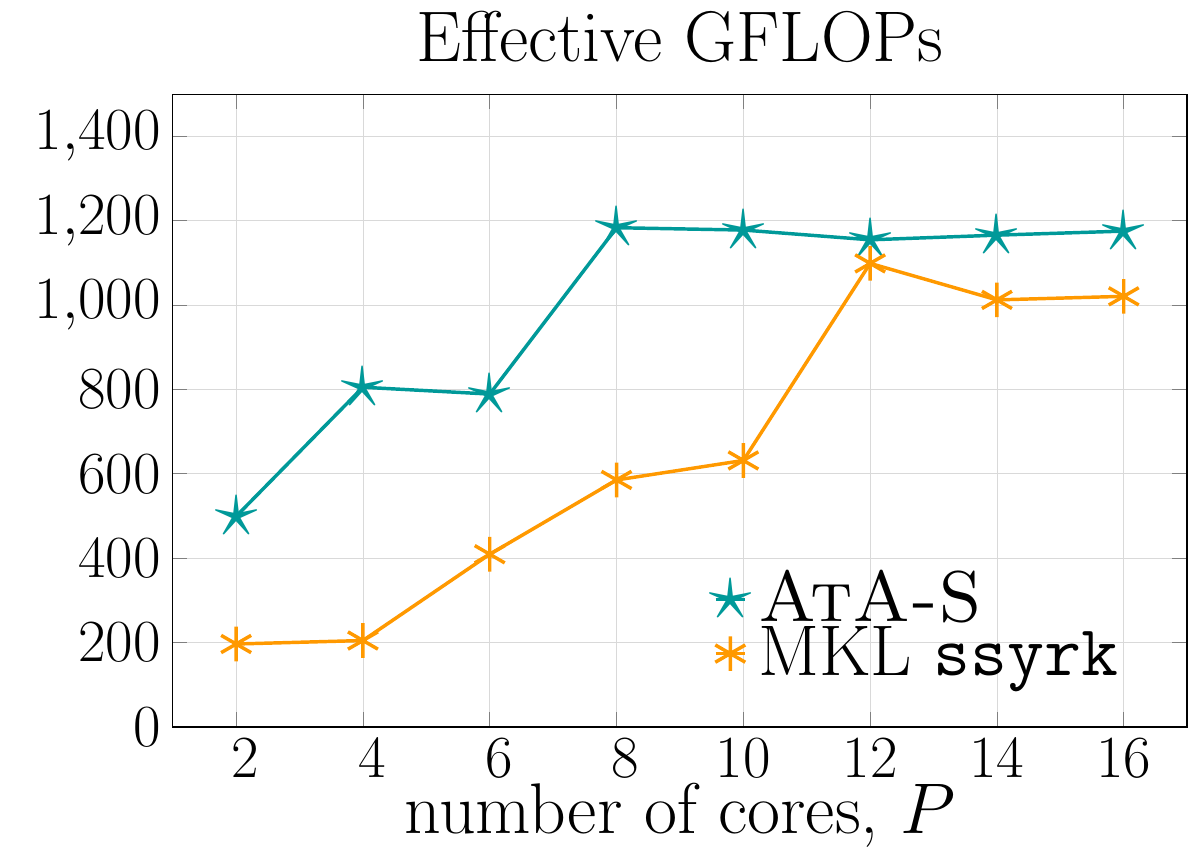}}
    
    \subfigure[Elapsed time, $\mA\in \mathbb{R}^{40K\times 40K}$.\label{fig:shar-40time}]{\includegraphics[width=0.49\columnwidth]{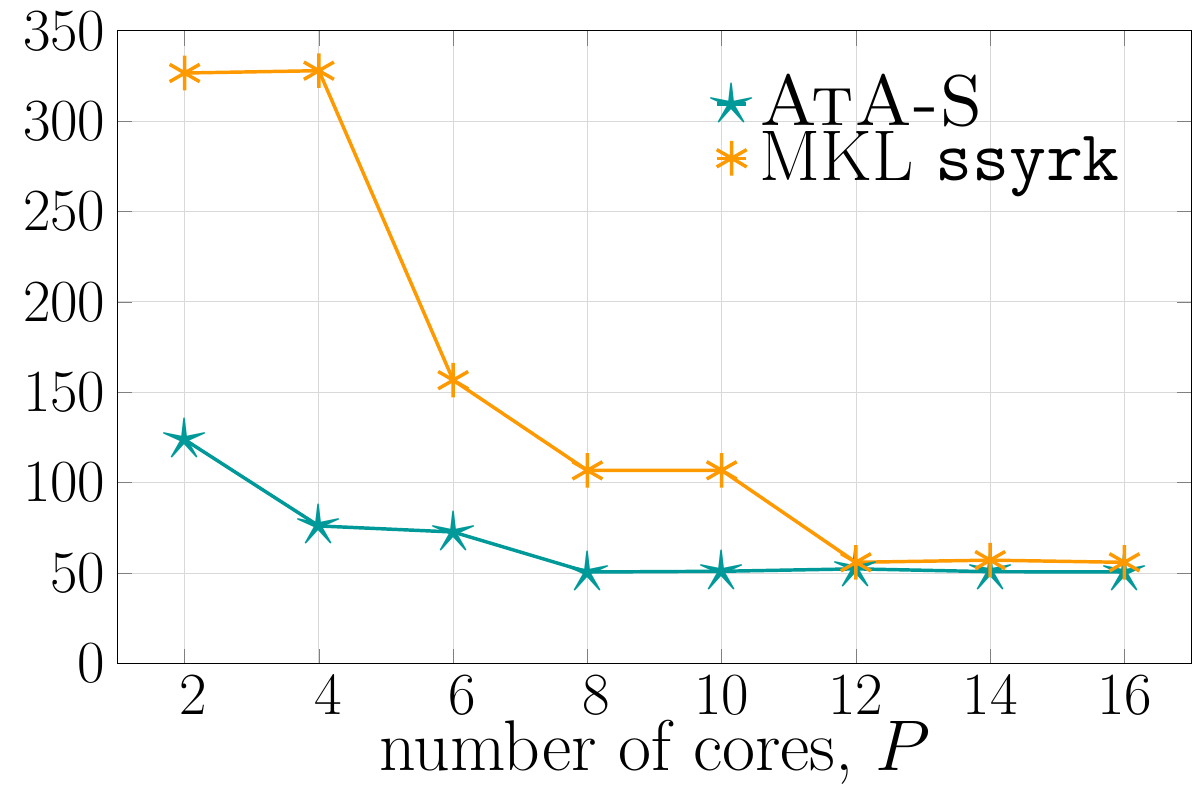}}
    \subfigure[EGs, $\mA\in \mathbb{R}^{40K\times 40K}$.\label{fig:shar-40gflop}]{\includegraphics[width=0.49\columnwidth]{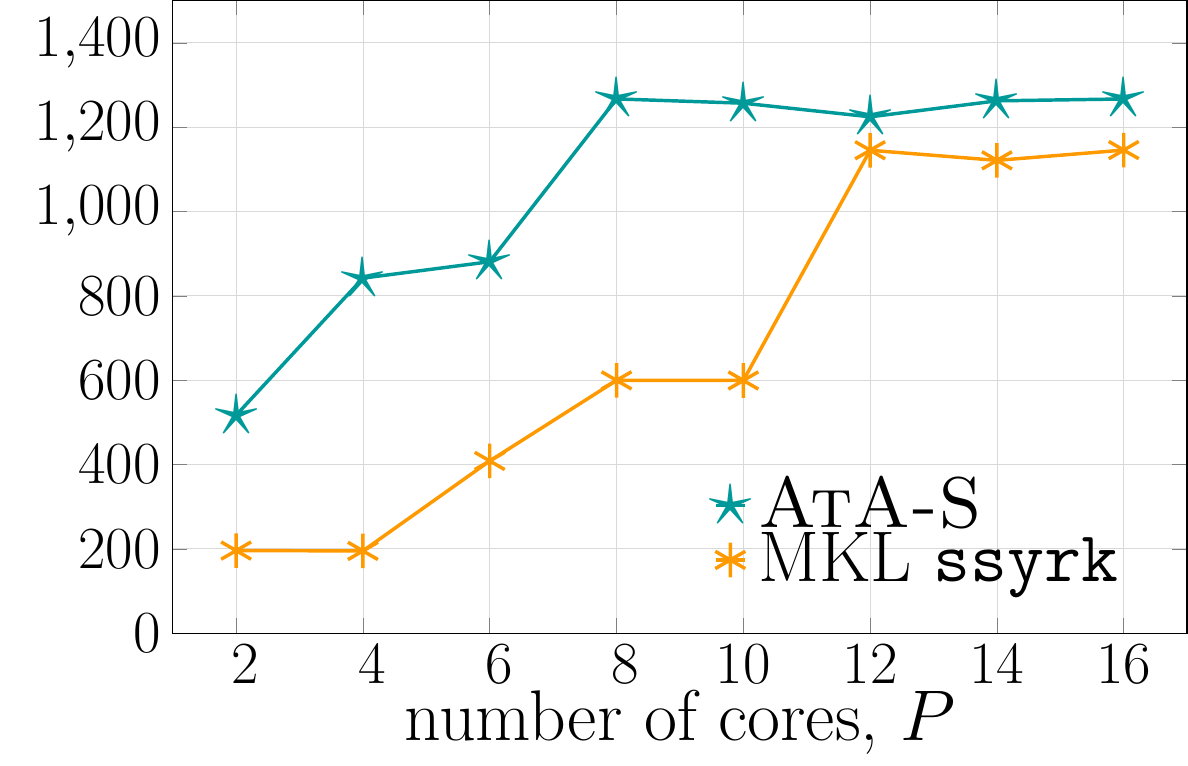}}

    \subfigure[Elapsed time, $\mA\in \mathbb{R}^{60K\times 5K}$.\label{fig:shar-605time}]{\includegraphics[width=0.49\columnwidth]{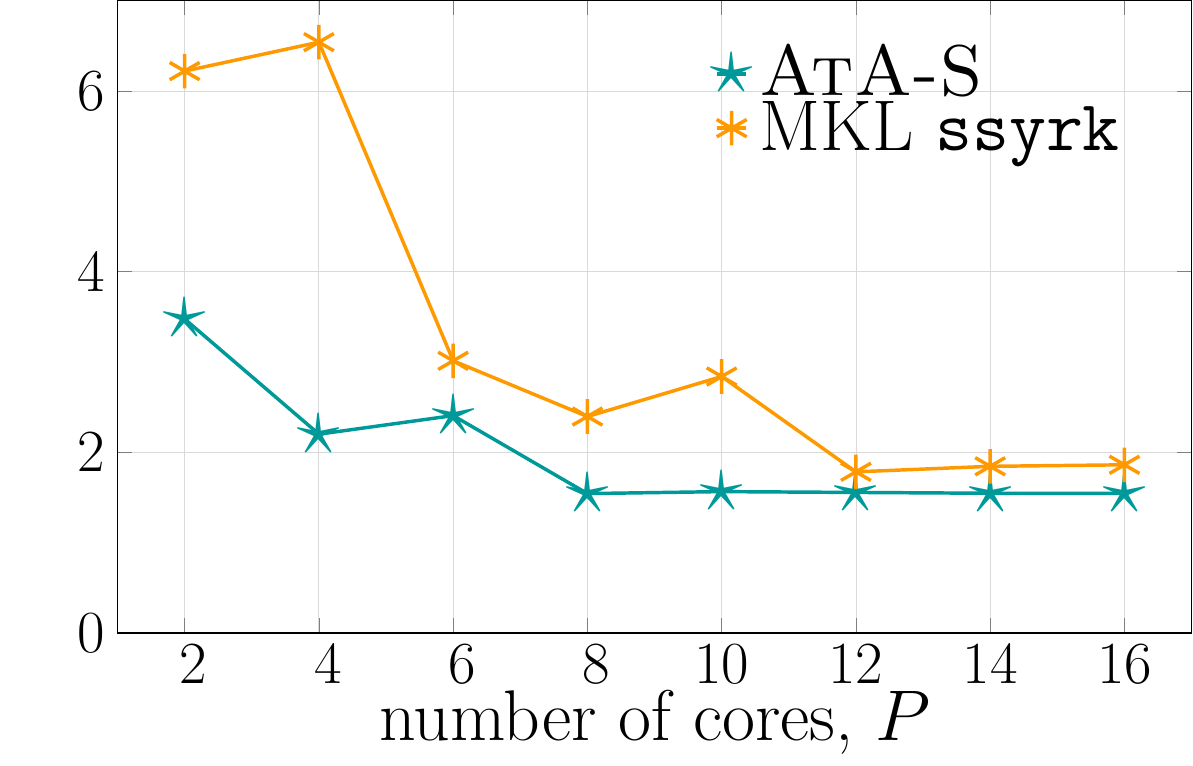}}
    \subfigure[EGs, $\mA\in \mathbb{R}^{60K\times 5K}$.\label{fig:shar-605gflop}]{\includegraphics[width=0.49\columnwidth]{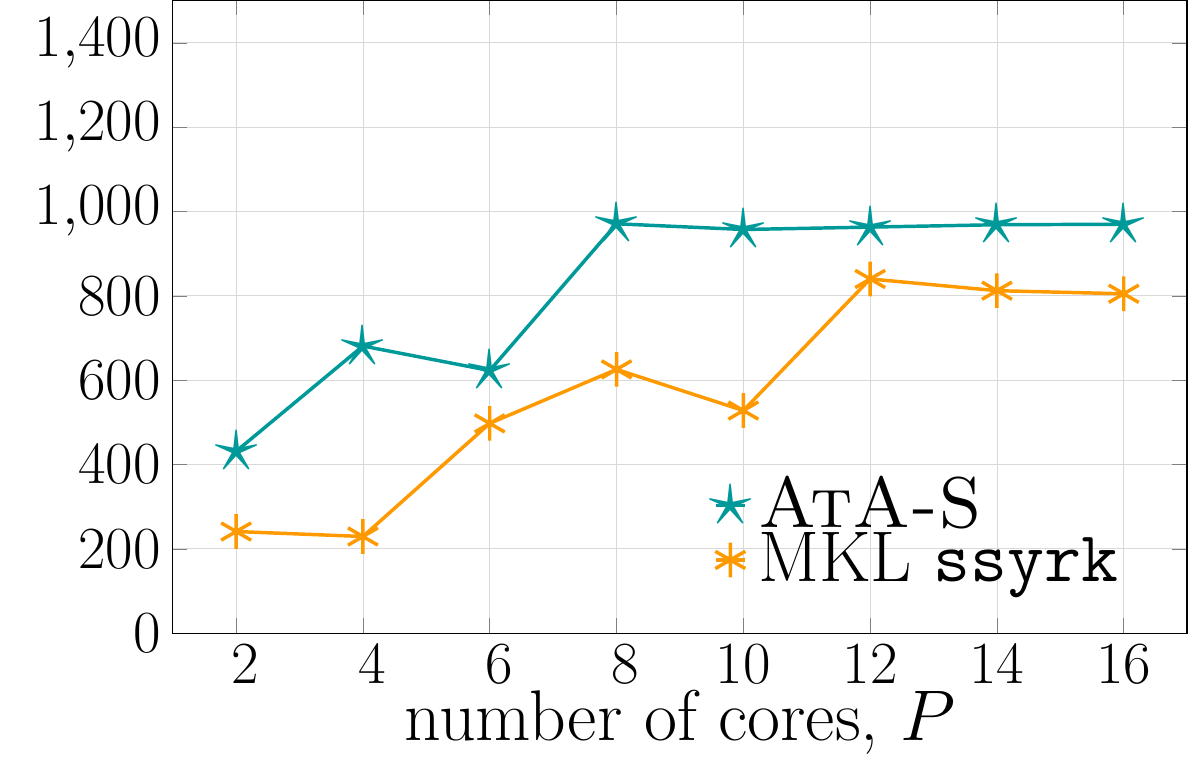}}
    
    \caption{Experimental results of \ATAS and Intel MKL \texttt{dsyrk} in terms of elapsed time in seconds (left column) and effective GFLOPs (right column), varying the number of available cores $P$ on fixed matrix sizes with a 16 threads configuration. }
    \label{fig:shared}
\end{figure}
From Figure~\ref{fig:shared}, we can observe that the performance of both methods stall when more than 8 cores are used. Indeed, multi-threaded MKL automatically chooses the optimal number of threads (in our architecture, this corresponds to 16 threads). For a fair comparison, we use the same setup in \ATAS. Performance scales with the number of available cores, but, when hyper-threading is enabled, 8 cores are enough to reach the 16-thread plateau. Therefore, performance cannot increase significantly for $P>8$.

\subsection{Distributed memory}
To complete our performance evaluation, we also compare our implementation for distributed architectures of \ATA, \ATAP, with fast distributed algorithms for matrix multiplication. We recall that \ATAP differs from standard methods for distributed matrix multiplication, as it does not perform computations on distributed matrices. Instead, in \ATAP the input matrix $\mA$ is only stored by the root process, $p_0$, that first distributes it among other processes cooperating to perform the $\mA^T\mA$ product, and then collects the partial result of each process to combine them. This approach makes our method unsuitable for distributed chains of operations, since for every operation, the matrix must be repeatedly scattered and gathered back, thus introducing communication overhead, but our results highlight that it is an efficient alternative for distributing single $\mA^T\mA$ operations. At the current state-of-the-art, there are a variety of methods for multiplying distributed matrices, but in the most recent literature there are three algorithms which stand out:


\begin{figure*}[!ht]
    \centering
    \subfigure[Elapsed time. $\mA,\mB \in \mathbb{R}^{10K\times 10K}$.\label{fig:time10kD}]{\includegraphics[width=0.5\columnwidth]{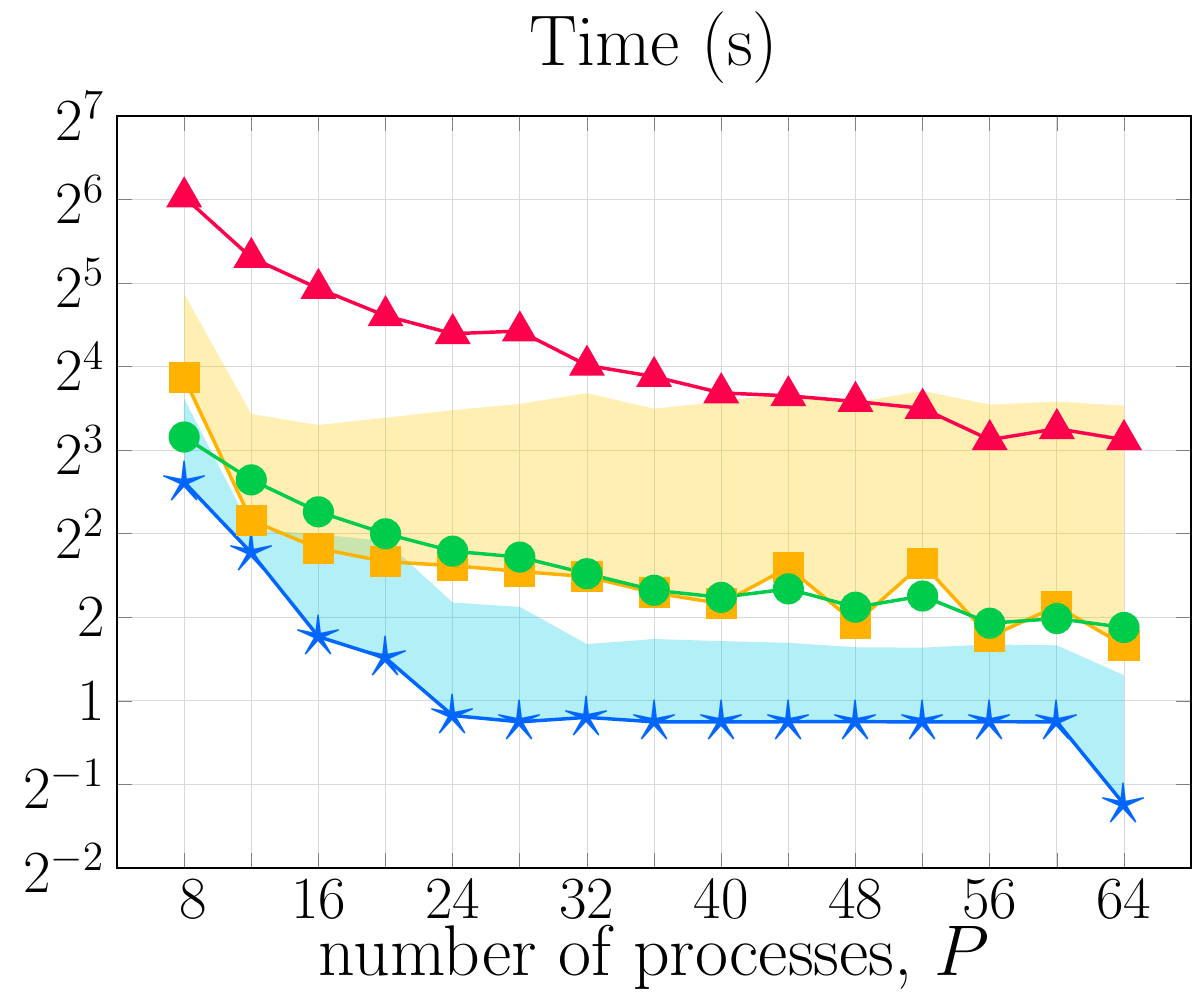}}
    \subfigure[EGs. $\mA,\mB \in \mathbb{R}^{10K\times 10K}$.\label{fig:EG10kD}]{\includegraphics[width=0.5\columnwidth]{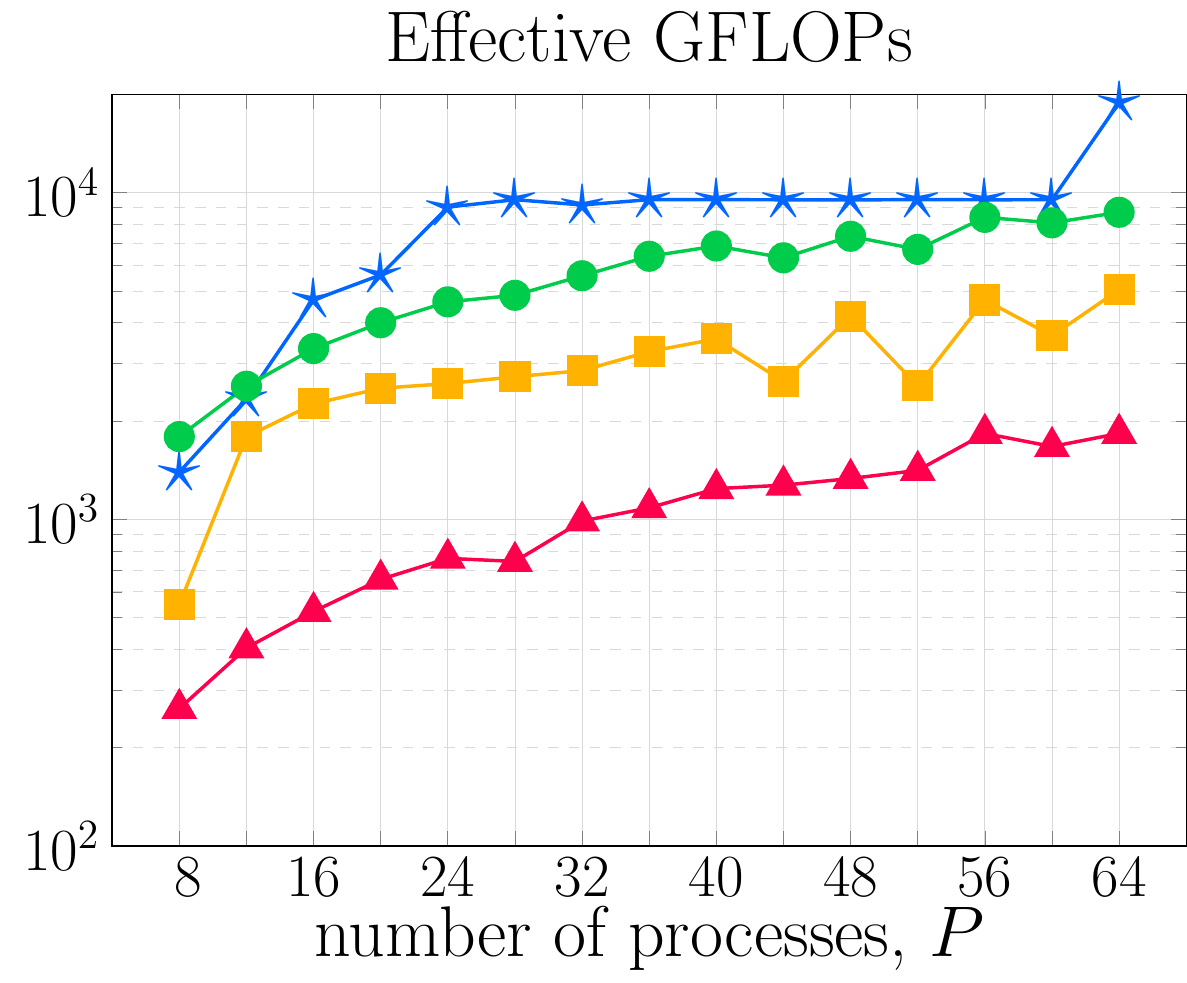}}
    \subfigure[TPP $\mA,\mB \in \mathbb{R}^{10K\times 10K}$.\label{fig:TPP10k}]{\includegraphics[width=0.5\columnwidth]{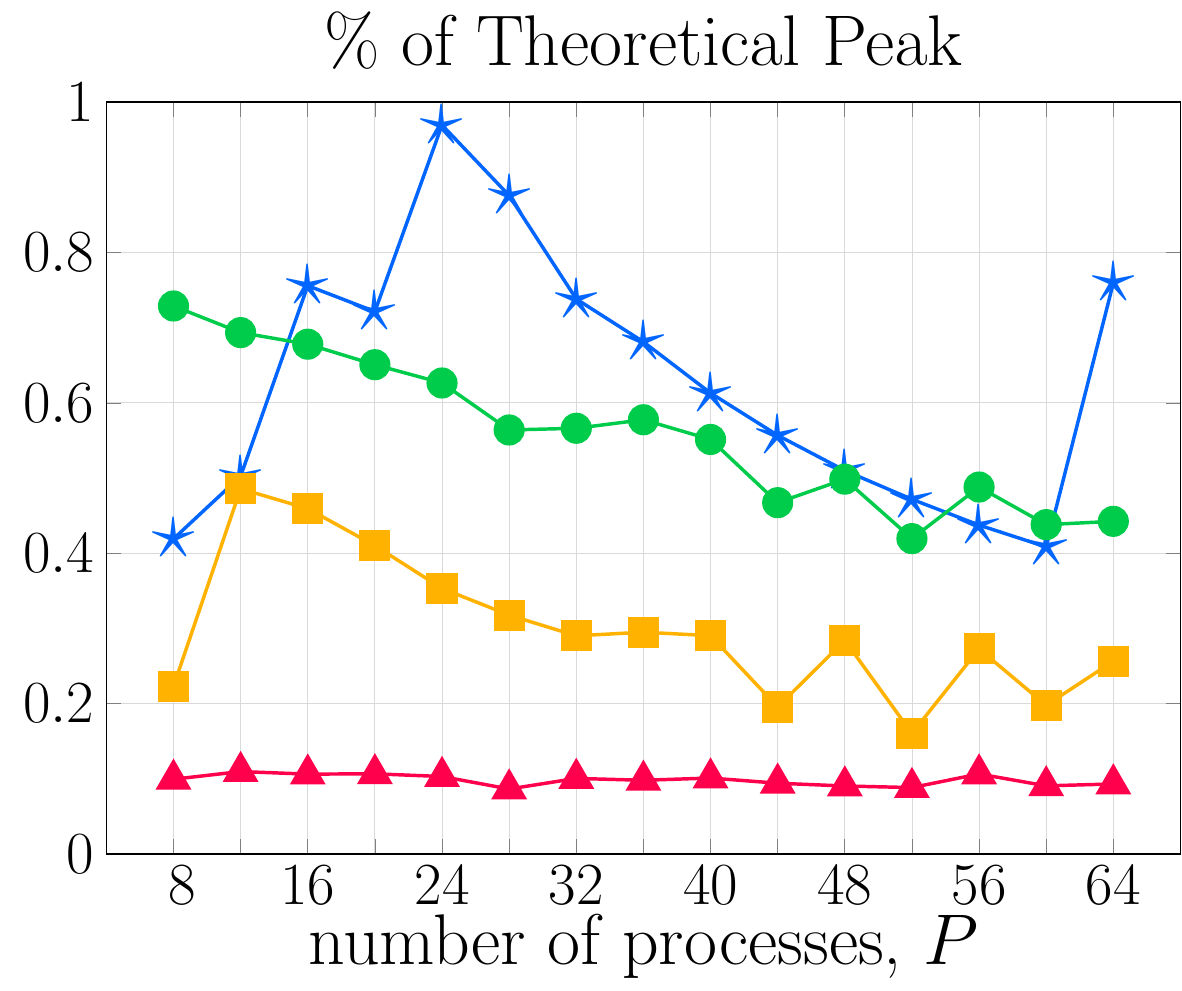}}
    \subfigure{\includegraphics[width=0.4\columnwidth]{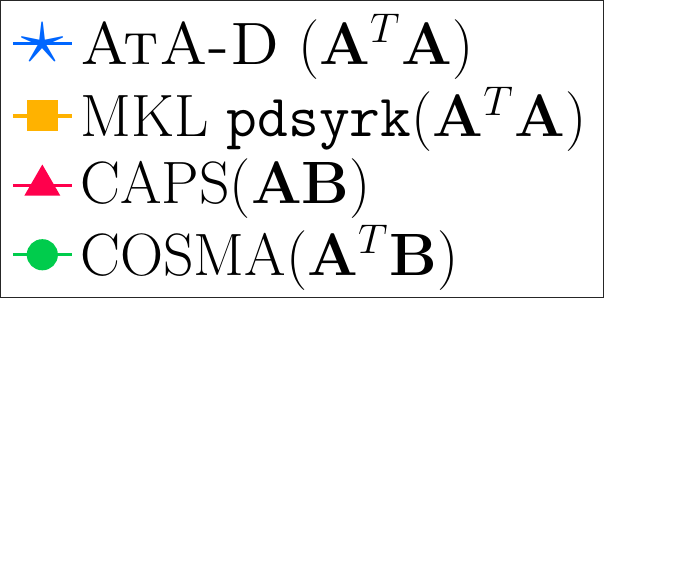}}
    
    \noindent
    \renewcommand{\thesubfigure}{(d)}
    \subfigure[Elapsed time. $\mA,\mB \in \mathbb{R}^{20K\times 20K}$.\label{fig:time20kD}]{\includegraphics[width=0.5\columnwidth]{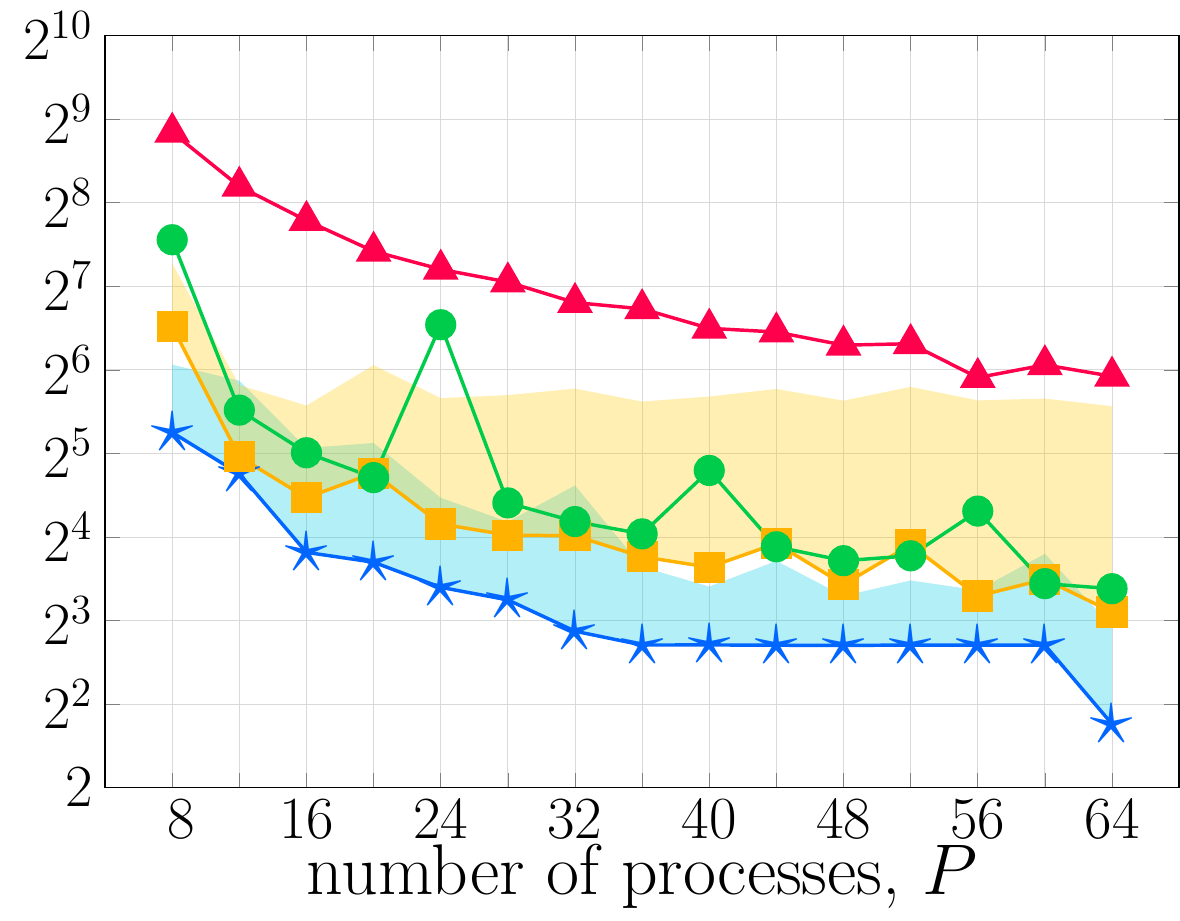}}
    \renewcommand{\thesubfigure}{(e)}
    \subfigure[EGs. $\mA,\mB \in \mathbb{R}^{20K\times 20K}$.\label{fig:EG20kD}]{\includegraphics[width=0.5\columnwidth]{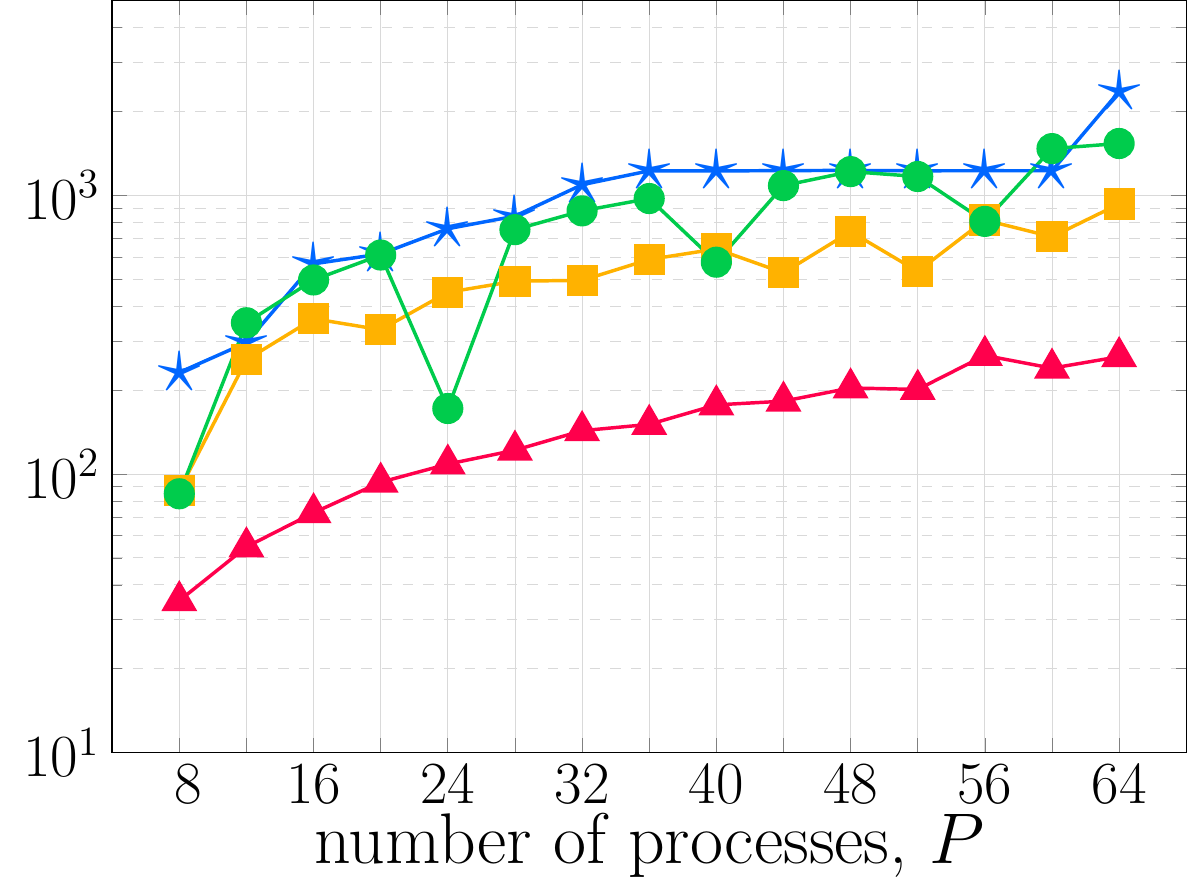}}
    \renewcommand{\thesubfigure}{(f)}
    \subfigure[TPP $\mA,\mB \in \mathbb{R}^{20K\times 20K}$.\label{fig:TPP20k}]{\includegraphics[width=0.5\columnwidth]{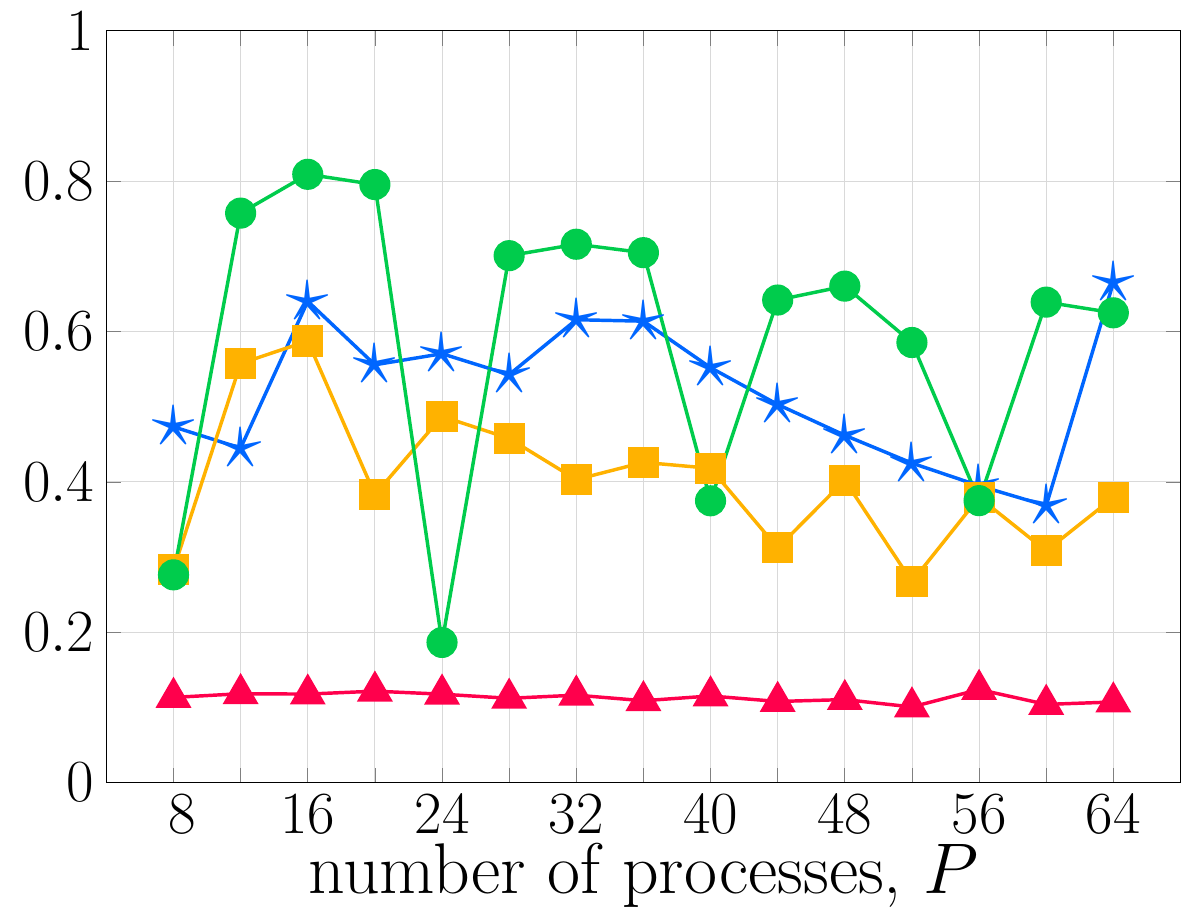}}
    \subfigure{\includegraphics[width=0.4\columnwidth]{imgs/distributed/Legend.pdf}}
    
    \noindent
    \renewcommand{\thesubfigure}{(g)}
    \subfigure[Elapsed time. $\mA,\mB \in \mathbb{R}^{60K\times 5K}$.\label{fig:time60kD}]{\includegraphics[width=0.5\columnwidth]{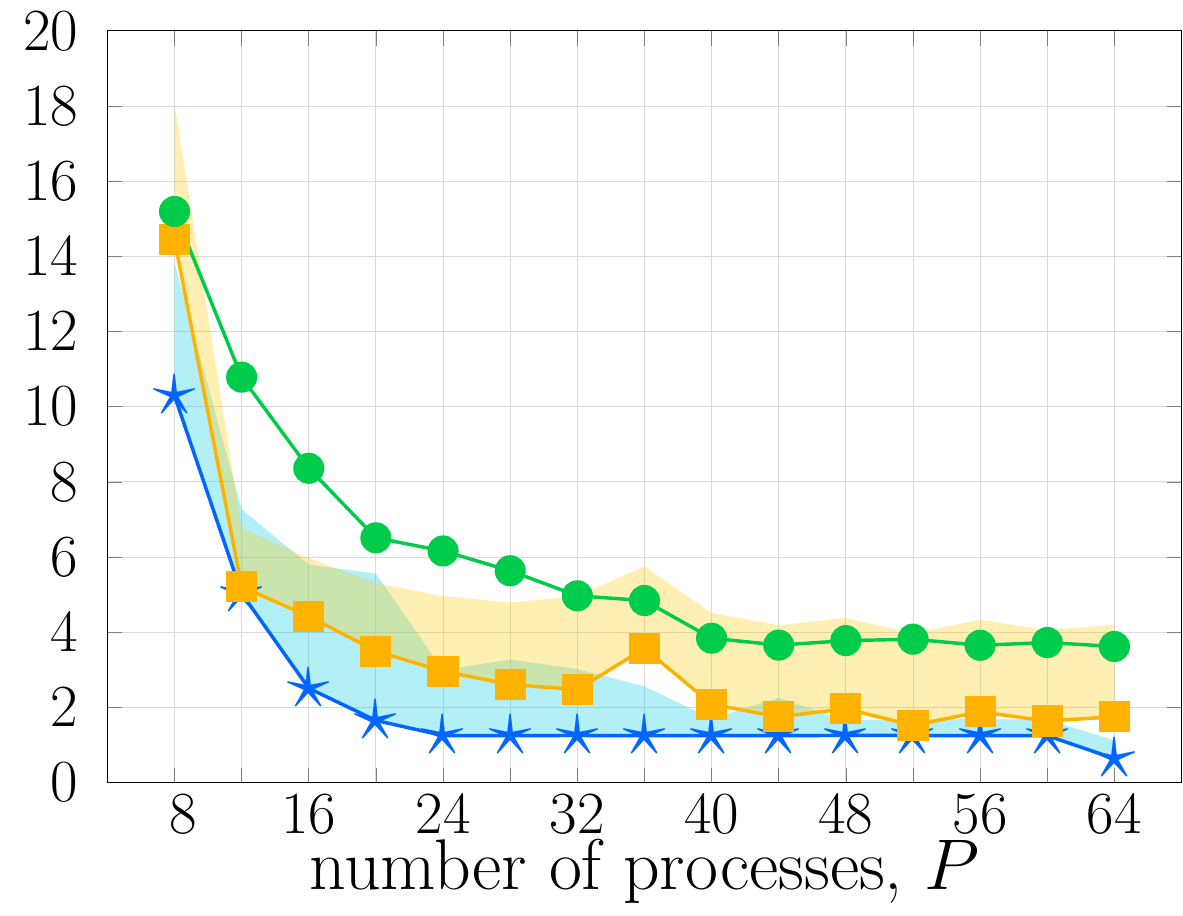}}
    \renewcommand{\thesubfigure}{(h)}
    \subfigure[EGs. $\mA,\mB \in \mathbb{R}^{60K\times 5K}$.\label{fig:EG60kD}]{\includegraphics[width=0.5\columnwidth]{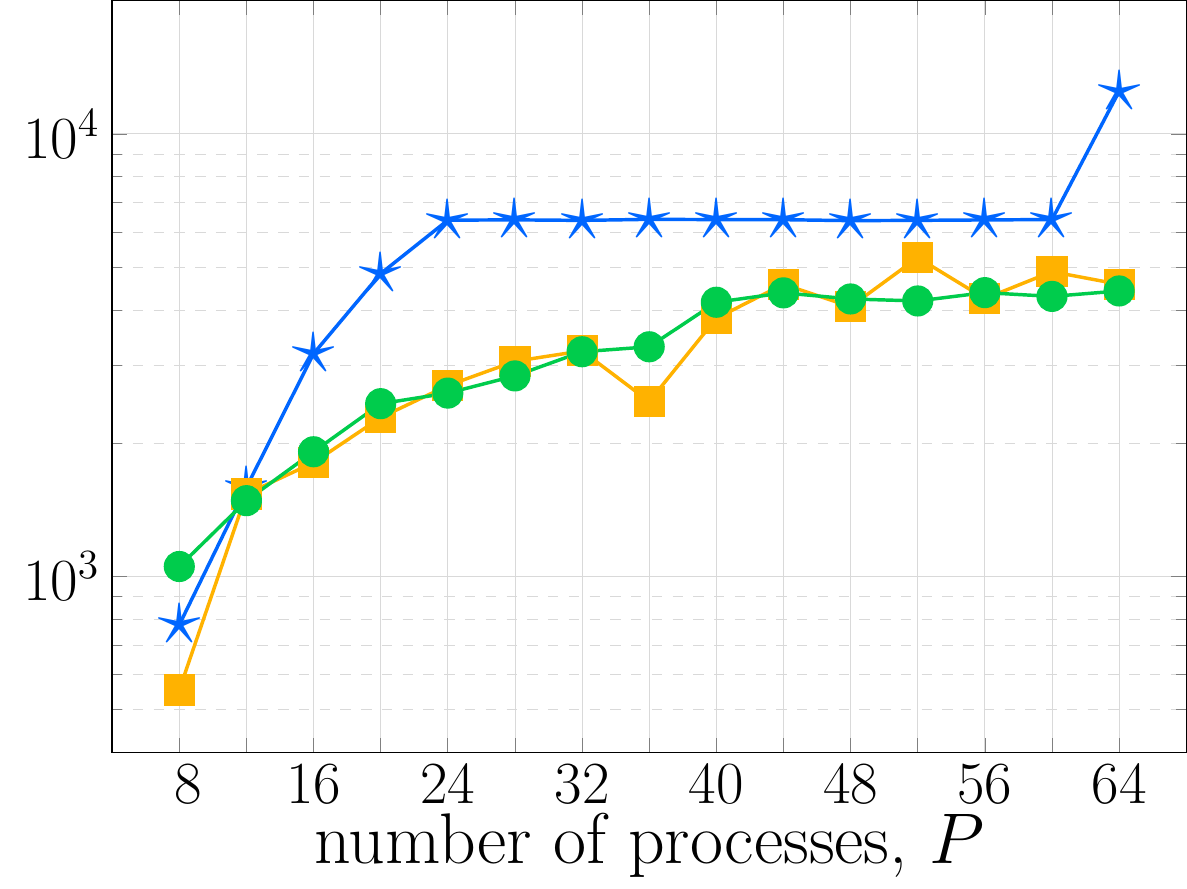}}
    \renewcommand{\thesubfigure}{(i)}
    \subfigure[TPP $\mA,\mB \in \mathbb{R}^{60K\times 5K}$.\label{fig:TPP60k}]{\includegraphics[width=0.5\columnwidth]{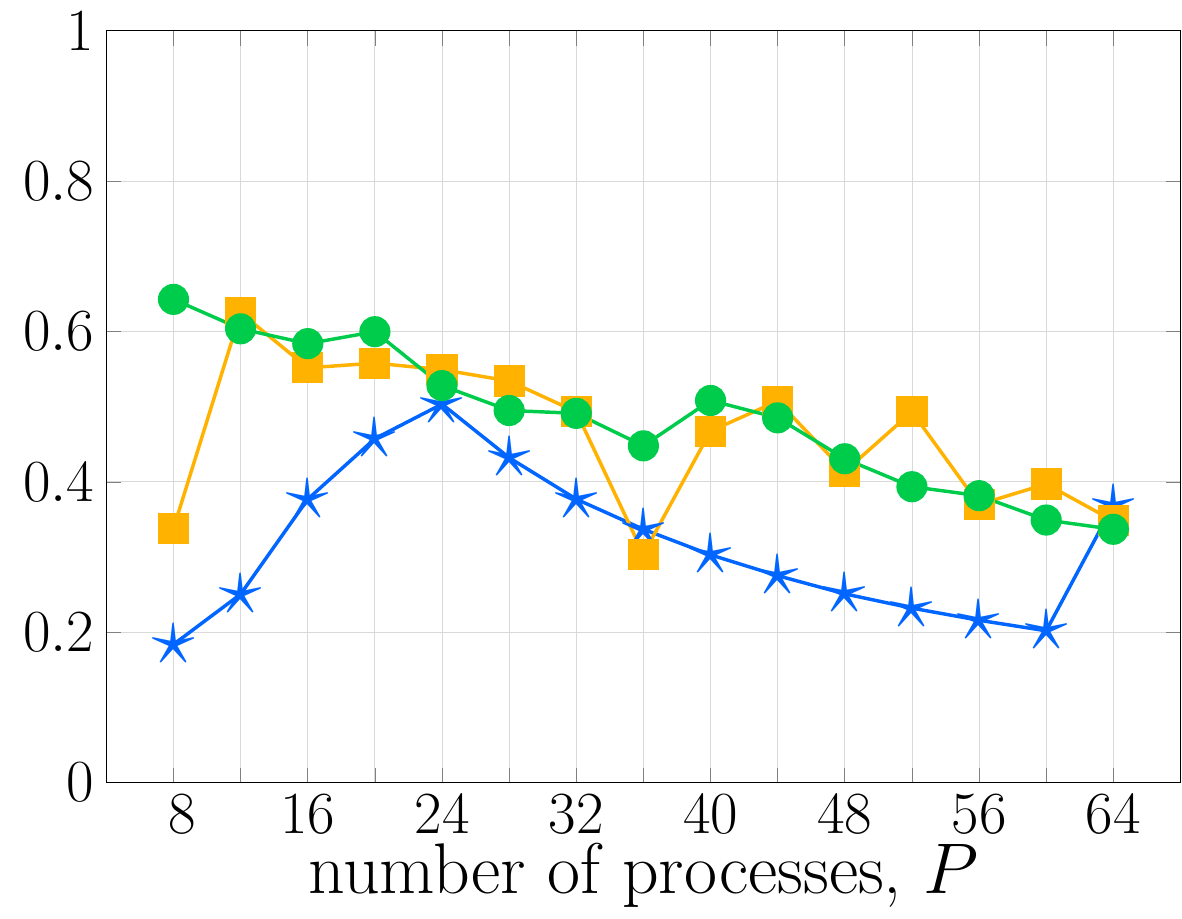}}
    \subfigure{\includegraphics[width=0.4\columnwidth]{imgs/distributed/Legend.pdf}}

    \caption{Experimental results of \ATAP, Intel MKL \texttt{pdsyrk}, CAPS and COSMA in terms of elapsed time in seconds (left column), effective GFLOPs (central column) and \% of theoretical peak (right column) varying the number of distributed processes $P$ on fixed matrix sizes.}
    \label{fig:distributed}
\end{figure*}

\begin{enumerate}
    \item Intel MKL ScaLapack \texttt{p?syrk}:
    the Intel Math Kernel libraries (MKL) provide optimized implementation of ScaLapack routines for high-performance dense Linear Algebra operations on distributed clusters. 
    In ScaLapack, distributed processes are organized in 2D grids of size $m_P\times n_P = P$.
    For each value of $P$, we set optimal $m_P$ and $n_P$ by calling \texttt{MPI\_Dims\_create}. We analyse the execution time required to perform the $\mA^T\mA$ matrix multiplication by the built-in ScaLapack function \texttt{pdsyrk}, and the time to retrieve the result of the operation. 
    
    \item CAPS\footnote{\url{https://github.com/lipshitz/CAPS/}}: the Communication-Optimal Parallel Algorithm for Strassen’s Matrix Multiplication \cite{ballard2012communication} is a distributed algorithm for general square matrix multiplications $\mA\mB$. Soon after CAPS, the same authors proposed CARMA \cite{demmel2013communication}, that also handles rectangular matrices. Nevertheless, it was not possible to test this method as it relies on Cilk Plus, a tool for parallel computing now marked as deprecated\footnote{https://www.cilkplus.org/, Last accessed 07-01-2021}. 
    
    \item COSMA\footnote{\url{https://github.com/eth-cscs/COSMA}}: differently from CAPS, this communication-optimal algorithm for general matrix multiplication does not rely on Strassen's algorithm, instead, it uses red-blue pebble game to precisely model the matrix-multiplication dependencies. In \cite{COSMA}, the authors show that COSMA outperforms all previously proposed frameworks for general matrix multiplication. It also works for multiplication on transposed matrices, and therefore we test it to perform $\mA^T\mB$ products. 
\end{enumerate}


To simulate massively distributed architectures, in our experiments, we reserve only one core per distributed process. 
As a consequence, each process has small memory availability (4GB RAM/core). The results of our experiments for the distributed-memory solution are shown in Figure~\ref{fig:distributed}. In Figures~\ref{fig:time10kD}, \ref{fig:time20kD} and \ref{fig:time60kD}, marked lines represent the compute time of all considered methods. The shaded areas above the curves describing \ATAP and \texttt{pdsyrk}  represent the additional time required for communication, i.e., for retrieving the resulting matrix to the root process. 
We consider two groups of square matrices, having size $10^4$ and $2\cdot 10^4$ (Figures~\ref{fig:time10kD}, \ref{fig:EG10kD}, \ref{fig:TPP10k} and \ref{fig:time20kD}, \ref{fig:EG20kD}, \ref{fig:TPP20k} respectively), and one set of tall matrices of size $6\cdot 10^4 \times 5\cdot 10^3$ (Figures~\ref{fig:time60kD}, \ref{fig:EG60kD}, \ref{fig:TPP60k}). Because CAPS does not operate on rectangular matrices, we could not test it on the latter set of experimental configurations. 
As we can observe from Figure~\ref{fig:distributed}, scalability of \ATAP is nonlinear and it rather follows an almost-stepwise trend with respect to $P$. This is a consequence of Equation~\ref{eq:numLevels}, that shows how some values of $P$ allow for a more effective and balanced workload between processes. This is evident for small values of $P$ (when a greater availability of processes weighs significantly on the workload of each process), as well as for $P=64$.  
Despite the different nature of the parallelism implemented in \ATAP with respect to the benchmark methods analysed in this section, our experiments corroborate the efficiency of the task distribution implemented in \ATAP. In Figures~\ref{fig:TPP10k},~\ref{fig:TPP20k} and \ref{fig:TPP60k} we show the percentage of theoretical peak performance (TPP) for all tested algorithms. We compute it as the effective GFLOPs over the theoretical performance peak of the nodes of our cluster. For all tested methods, the effective GFLOPs are computed as in Equation~\ref{eq:effGFLOPs}, (as
those reported in Figures~\ref{fig:EG10kD}, \ref{fig:EG20kD}, \ref{fig:EG60kD}), except for  \ATAP, for which we now use the complexity of \ATA (Equation~\ref{eq:ATACC}).
Regarding the percentage of theoretical peak, we can see how our algorithm has comparable behaviour with respect to the other solutions on square matrices, but it performs worse on the rectangular case. The high performance of our method relies on careful ordering and placement of highly optimized BLAS routines. However, especially when working on tall matrices, we need to perform several calls to BLAS Level 1 routines (i.e., to compute intermediate sums both in \ATA and \Strassen) and system calls (i.e., memory copies) on very short rows. 
This leads to more memory accesses and poorer vectorization capability than dealing with the same amount of data, distributed in fewer, longer rows, would entail. As a consequence, the overall performance with respect to the theoretical peak is worsened. Furthermore, in Figures~\ref{fig:TPP10k},~\ref{fig:TPP20k} and \ref{fig:TPP60k} we see that \ATAP loses a bit of efficiency on larger matrices. This is due to the fact that each process calling the \Strassen routine needs to allocate $\nicefrac{3}{2}n^2$ space of memory, hence memory handling slows down the entire process. In addition, also in the last column of Figure~\ref{fig:distributed},  we can observe performance peaks after slow degradations as we did in the first two. We stress that this is a consequence of the fact that for some values of $P$, workload among processes is distributed more efficiently (see Equation~\ref{eq:numLevels}). As a matter of fact, the computational complexity of \ATAP decreases exponentially at each level, but the number of levels increases logarithmically with the number of processes, 
$P$. Therefore for numbers of processes that result in the same number of parallel levels, the improvement is less appreciable. \newline \indent Finally, in order to study the scalability of \ATAP with respect to \ATAS, and to validate the possibility of integrating the two methods, we compare \ATAS and \ATAP on very large matrices of increasing size and report results in Table~\ref{tab:shared-vs-distr}. \ATAS works on 16 cores with 16 threads,  whereas \ATAP works on 6 distributed nodes, each equipped with 16 cores, for a total of 96 cores. Each node executes a distributed process calling either \ATAS for $\mA^T\mA$-type products, or multi-threaded MKL \texttt{dgemm} for $\mA^T\mB$-type multiplications. The times reported in Table~\ref{tab:shared-vs-distr} for \ATAP also include communication time (for distributing data and collecting results). Speed-up is computed as $\nicefrac{T_{SM}}{T_{DM}}$, where $T_{SM}$ and $T_{DM}$ are the execution times of the shared and the distributed-memory algorithms, respectively.  
In accordance with our computational and communication cost analysis (Section~\ref{sec:CCatap}), the speed-up of \ATAP over \ATAS increases when the size of the input matrix increases, as the computation cost overwhelms the communication overhead. Furthermore, the shared-memory implementation suffers when considering larger matrices, since frequent memory access slows down execution (two $60K \times 60K$ matrices require 57 GB out of the 64 GB available on the test machine), consequently decreasing performance. This is highlighted by the results of Table~\ref{tab:shared-vs-distr}, where we can observe that $60K \times 60K$ matrices require high computation time, dominated by the time for memory management.

\begin{table}[t!]
    \centering
    \begin{tabular}{ccccc}
$n$
&
SM (16 cores)
&
DM (96 cores)
&
Speed-up
\\
\hline
$30\mathrm{K}$
&
45.16 s
&
21.24 s
&
2.13
\\
$40\mathrm{K}$
&
106.19 s
&
43.96 s
&
2.42
\\
$50\mathrm{K}$
&
221.63 s
&
81.77 s
&
2.71
\\
$60\mathrm{K}$
&
863.32 s
&
129.08 s
&
6.69
\\
\end{tabular}
    \caption{Shared memory (SM) vs distributed memory (DM) $\mA^T\mA$ implementation on large square $n \times n$ matrices.} 
    \label{tab:shared-vs-distr}
    \vspace{-.5cm}
\end{table}

\section{Conclusions}
We propose \ATA, an algorithm for the $\mA^T\mA$ product, that reduces the computational complexity of commonly employed algorithms, and that is conveniently implementable in practice on matrices defined on arbitrary domains and of any size and aspect ratio. The computational cost of \ATA benefits from the fast matrix multiplication introduced by Strassen's algorithm, and is cache-oblivious. We show that \ATA can be efficiently implemented in shared and distributed memory environments. In the shared memory implementation of \ATA, tasks are assigned to parallel threads so that they work in perfect parallelism. Our theoretical analysis is supported by experiments that prove the excellent performance of our implementations in comparison with state-of-the-art counterparts.

\begin{acks}
This work has been partially supported through the ERC Starting Grant n. 802554 (SPECGEO).
\end{acks}

\bibliographystyle{ACM-Reference-Format}
\bibliography{biblio}


\begin{thebibliography}{38}


\ifx \showCODEN    \undefined \def \showCODEN     #1{\unskip}     \fi
\ifx \showDOI      \undefined \def \showDOI       #1{#1}\fi
\ifx \showISBNx    \undefined \def \showISBNx     #1{\unskip}     \fi
\ifx \showISBNxiii \undefined \def \showISBNxiii  #1{\unskip}     \fi
\ifx \showISSN     \undefined \def \showISSN      #1{\unskip}     \fi
\ifx \showLCCN     \undefined \def \showLCCN      #1{\unskip}     \fi
\ifx \shownote     \undefined \def \shownote      #1{#1}          \fi
\ifx \showarticletitle \undefined \def \showarticletitle #1{#1}   \fi
\ifx \showURL      \undefined \def \showURL       {\relax}        \fi
\providecommand\bibfield[2]{#2}
\providecommand\bibinfo[2]{#2}
\providecommand\natexlab[1]{#1}
\providecommand\showeprint[2][]{arXiv:#2}

\bibitem[\protect\citeauthoryear{Ballard, Demmel, Holtz, Lipshitz, and
  Schwartz}{Ballard et~al\mbox{.}}{2012}]%
        {ballard2012communication}
\bibfield{author}{\bibinfo{person}{G. Ballard}, \bibinfo{person}{J. Demmel},
  \bibinfo{person}{O. Holtz}, \bibinfo{person}{B. Lipshitz}, {and}
  \bibinfo{person}{O. Schwartz}.} \bibinfo{year}{2012}\natexlab{}.
\newblock \showarticletitle{Communication-optimal parallel algorithm for
  strassen's matrix multiplication}. In \bibinfo{booktitle}{\emph{Proc. 24th
  ACM Symp. Parallelism in Algorithms and Architectures (SPAA)}}.
  \bibinfo{pages}{193--204}.
\newblock


\bibitem[\protect\citeauthoryear{Ballard, Demmel, Holtz, and Schwartz}{Ballard
  et~al\mbox{.}}{2011}]%
        {BDHS-SPAA11}
\bibfield{author}{\bibinfo{person}{G. Ballard}, \bibinfo{person}{J. Demmel},
  \bibinfo{person}{O. Holtz}, {and} \bibinfo{person}{O. Schwartz}.}
  \bibinfo{year}{2011}\natexlab{}.
\newblock \showarticletitle{Graph expansion and communication costs of fast
  matrix multiplication}. In \bibinfo{booktitle}{\emph{Proc. 23rd ACM Symp.
  Parallelism in Algorithms and Architectures (SPAA)}}. \bibinfo{pages}{1--12}.
\newblock


\bibitem[\protect\citeauthoryear{Ballard, Demmel, Holtz, and Schwartz}{Ballard
  et~al\mbox{.}}{2013}]%
        {ballard2013graph}
\bibfield{author}{\bibinfo{person}{G. Ballard}, \bibinfo{person}{J. Demmel},
  \bibinfo{person}{O. Holtz}, {and} \bibinfo{person}{O. Schwartz}.}
  \bibinfo{year}{2013}\natexlab{}.
\newblock \showarticletitle{Graph expansion and communication costs of fast
  matrix multiplication}.
\newblock \bibinfo{journal}{\emph{Journal of the ACM (JACM)}}
  \bibinfo{volume}{59}, \bibinfo{number}{6} (\bibinfo{year}{2013}),
  \bibinfo{pages}{1--23}.
\newblock


\bibitem[\protect\citeauthoryear{Benson and Ballard}{Benson and
  Ballard}{2015}]%
        {Benson-PPOPP15}
\bibfield{author}{\bibinfo{person}{A.R. Benson} {and} \bibinfo{person}{G.
  Ballard}.} \bibinfo{year}{2015}\natexlab{}.
\newblock \showarticletitle{A Framework for Practical Parallel Fast Matrix
  Multiplication}. In \bibinfo{booktitle}{\emph{Proc. 20th ACM SIGPLAN PPoPP}}.
  \bibinfo{pages}{42--53}.
\newblock


\bibitem[\protect\citeauthoryear{Brent}{Brent}{1970a}]%
        {brent-TR70}
\bibfield{author}{\bibinfo{person}{R.~P. Brent}.}
  \bibinfo{year}{1970}\natexlab{a}.
\newblock \bibinfo{title}{Algorithms for matrix multiplication}.
\newblock \bibinfo{howpublished}{Tech. Rep. TR-CS-70-157, Stanford University}.
\newblock


\bibitem[\protect\citeauthoryear{Brent}{Brent}{1970b}]%
        {brent-NM70}
\bibfield{author}{\bibinfo{person}{R.~P. Brent}.}
  \bibinfo{year}{1970}\natexlab{b}.
\newblock \showarticletitle{Error analysis of algorithms for matrix
  multiplication and triangular decomposition using Winograd’s identity}.
\newblock \bibinfo{journal}{\emph{Numer. Math.}}  \bibinfo{volume}{16}
  (\bibinfo{year}{1970}), \bibinfo{pages}{145–156}.
\newblock


\bibitem[\protect\citeauthoryear{Charara, Keyes, and Ltaief}{Charara
  et~al\mbox{.}}{2019}]%
        {charara2019batched}
\bibfield{author}{\bibinfo{person}{A. Charara}, \bibinfo{person}{D. Keyes},
  {and} \bibinfo{person}{H. Ltaief}.} \bibinfo{year}{2019}\natexlab{}.
\newblock \showarticletitle{Batched triangular dense linear algebra kernels for
  very small matrix sizes on GPUs}.
\newblock \bibinfo{journal}{\emph{ACM Trans Math. Softw. (TOMS)}}
  \bibinfo{volume}{45}, \bibinfo{number}{2} (\bibinfo{year}{2019}),
  \bibinfo{pages}{1--28}.
\newblock


\bibitem[\protect\citeauthoryear{Charara, Ltaief, and Keyes}{Charara
  et~al\mbox{.}}{2016}]%
        {charara2016redesigning}
\bibfield{author}{\bibinfo{person}{A. Charara}, \bibinfo{person}{H. Ltaief},
  {and} \bibinfo{person}{D. Keyes}.} \bibinfo{year}{2016}\natexlab{}.
\newblock \showarticletitle{Redesigning triangular dense matrix computations on
  GPUs}. In \bibinfo{booktitle}{\emph{Euro-Par}}. Springer,
  \bibinfo{pages}{477--489}.
\newblock


\bibitem[\protect\citeauthoryear{{Coppersmith} and Winograd}{{Coppersmith} and
  Winograd}{1987}]%
        {coppersmith1987matrix}
\bibfield{author}{\bibinfo{person}{D. {Coppersmith}} {and} \bibinfo{person}{S.
  Winograd}.} \bibinfo{year}{1987}\natexlab{}.
\newblock \showarticletitle{Matrix multiplication via arithmetic progressions}.
  In \bibinfo{booktitle}{\emph{Proc. 19th ACM Symp. Theory of Computing
  (STOC)}}. \bibinfo{pages}{1--6}.
\newblock


\bibitem[\protect\citeauthoryear{D'Alberto and Nicolau}{D'Alberto and
  Nicolau}{2007}]%
        {d2007adaptive}
\bibfield{author}{\bibinfo{person}{P. D'Alberto} {and} \bibinfo{person}{A.
  Nicolau}.} \bibinfo{year}{2007}\natexlab{}.
\newblock \showarticletitle{Adaptive Strassen's matrix multiplication}. In
  \bibinfo{booktitle}{\emph{Proc. 21st Int. Conf. on Supercomputing}}. ACM,
  \bibinfo{pages}{284--292}.
\newblock


\bibitem[\protect\citeauthoryear{Demmel, Eliahu, Fox, Kamil, Lipshitz, O., and
  Spillinger}{Demmel et~al\mbox{.}}{2013}]%
        {demmel2013communication}
\bibfield{author}{\bibinfo{person}{J. Demmel}, \bibinfo{person}{D. Eliahu},
  \bibinfo{person}{A. Fox}, \bibinfo{person}{S. Kamil}, \bibinfo{person}{B.
  Lipshitz}, \bibinfo{person}{O.~Schwartz O.}, {and}
  \bibinfo{person}{Spillinger}.} \bibinfo{year}{2013}\natexlab{}.
\newblock \showarticletitle{Communication-optimal parallel recursive
  rectangular matrix multiplication}. In \bibinfo{booktitle}{\emph{IEEE 27th
  Int. Symp. on Parallel and Distributed Processing}}. IEEE,
  \bibinfo{pages}{261--272}.
\newblock


\bibitem[\protect\citeauthoryear{Desprez and Suter}{Desprez and Suter}{2004}]%
        {Desprez-04}
\bibfield{author}{\bibinfo{person}{F. Desprez} {and} \bibinfo{person}{F.
  Suter}.} \bibinfo{year}{2004}\natexlab{}.
\newblock \showarticletitle{Impact of mixed-parallelism on parallel
  implementations of the Strassen and Winograd matrix multiplication
  algorithms}.
\newblock \bibinfo{journal}{\emph{Concurr. Comput.: Pract. Exper.}}
  \bibinfo{volume}{16}, \bibinfo{number}{8} (\bibinfo{year}{2004}),
  \bibinfo{pages}{771–797}.
\newblock


\bibitem[\protect\citeauthoryear{Dumas, Pernet, and Sedoglavic}{Dumas
  et~al\mbox{.}}{2020}]%
        {dumas2020fast}
\bibfield{author}{\bibinfo{person}{J. Dumas}, \bibinfo{person}{C. Pernet},
  {and} \bibinfo{person}{A. Sedoglavic}.} \bibinfo{year}{2020}\natexlab{}.
\newblock \showarticletitle{On Fast Multiplication of a Matrix by Its
  Transpose}. In \bibinfo{booktitle}{\emph{Proc. 45th Int. Symp. Symbolic and
  Algebraic Computation}} (Kalamata, Greece) \emph{(\bibinfo{series}{ISSAC
  '20})}. \bibinfo{publisher}{Association for Computing Machinery},
  \bibinfo{address}{New York, NY, USA}, \bibinfo{pages}{162–169}.
\newblock
\showISBNx{9781450371001}
\urldef\tempurl%
\url{https://doi.org/10.1145/3373207.3404021}
\showDOI{\tempurl}


\bibitem[\protect\citeauthoryear{Eliahu, Spillinger, Fox, and Demmel}{Eliahu
  et~al\mbox{.}}{2015}]%
        {eliahu2015frpa}
\bibfield{author}{\bibinfo{person}{D. Eliahu}, \bibinfo{person}{O. Spillinger},
  \bibinfo{person}{A. Fox}, {and} \bibinfo{person}{J. Demmel}.}
  \bibinfo{year}{2015}\natexlab{}.
\newblock \bibinfo{booktitle}{\emph{Frpa: A framework for recursive parallel
  algorithms}}.
\newblock \bibinfo{type}{{T}echnical {R}eport} UCB/EECS-2015-28.
  \bibinfo{institution}{EECS Department, University of California, Berkeley}.
\newblock


\bibitem[\protect\citeauthoryear{Elmroth, Gustavson, Jonsson, and
  K{\aa}gstr{\"o}m}{Elmroth et~al\mbox{.}}{2004}]%
        {elmroth2004recursive}
\bibfield{author}{\bibinfo{person}{E. Elmroth}, \bibinfo{person}{F. Gustavson},
  \bibinfo{person}{I. Jonsson}, {and} \bibinfo{person}{B. K{\aa}gstr{\"o}m}.}
  \bibinfo{year}{2004}\natexlab{}.
\newblock \showarticletitle{Recursive blocked algorithms and hybrid data
  structures for dense matrix library software}.
\newblock \bibinfo{journal}{\emph{SIAM review}} \bibinfo{volume}{46},
  \bibinfo{number}{1} (\bibinfo{year}{2004}), \bibinfo{pages}{3--45}.
\newblock


\bibitem[\protect\citeauthoryear{for Intel{\textregistered} Math Kernel
  Library~C}{for Intel{\textregistered} Math Kernel Library~C}{2019}]%
        {intelMKL}
\bibfield{author}{\bibinfo{person}{Developer~Reference for
  Intel{\textregistered} Math Kernel Library~C}.}
  \bibinfo{year}{2019}\natexlab{}.
\newblock  (\bibinfo{year}{2019}).
\newblock
\urldef\tempurl%
\url{https://software.intel.com/en-us/download/developer-reference-for-intel-math-kernel-library-c}
\showURL{%
\tempurl}


\bibitem[\protect\citeauthoryear{Frigo, Leiserson, Prokop, and
  Ramachandran}{Frigo et~al\mbox{.}}{1999}]%
        {frigo1999cache}
\bibfield{author}{\bibinfo{person}{M. Frigo}, \bibinfo{person}{C.~E.
  Leiserson}, \bibinfo{person}{H. Prokop}, {and} \bibinfo{person}{S.
  Ramachandran}.} \bibinfo{year}{1999}\natexlab{}.
\newblock \showarticletitle{Cache-oblivious algorithms}. In
  \bibinfo{booktitle}{\emph{40th Symp. Foundations of Computer Science
  (FOCS)}}. IEEE, \bibinfo{pages}{285--297}.
\newblock


\bibitem[\protect\citeauthoryear{Gall}{Gall}{2014}]%
        {le2014powers}
\bibfield{author}{\bibinfo{person}{F.~Le Gall}.}
  \bibinfo{year}{2014}\natexlab{}.
\newblock \showarticletitle{Powers of tensors and fast matrix multiplication}.
  In \bibinfo{booktitle}{\emph{Proc. 39th Int. Symp. Symbolic and Algebraic
  Computation}}. \bibinfo{pages}{296--303}.
\newblock


\bibitem[\protect\citeauthoryear{Grayson, Shah, and van~de Geijn}{Grayson
  et~al\mbox{.}}{1995}]%
        {Grayson-PPL95}
\bibfield{author}{\bibinfo{person}{B. Grayson}, \bibinfo{person}{A. Shah},
  {and} \bibinfo{person}{R. van~de Geijn}.} \bibinfo{year}{1995}\natexlab{}.
\newblock \showarticletitle{A high performance parallel Strassen
  implementation}.
\newblock \bibinfo{journal}{\emph{Parallel Processing Letters}}
  \bibinfo{volume}{6} (\bibinfo{year}{1995}), \bibinfo{pages}{3--12}.
\newblock


\bibitem[\protect\citeauthoryear{Higham}{Higham}{1990}]%
        {Higham-TOMS90}
\bibfield{author}{\bibinfo{person}{N.J. Higham}.}
  \bibinfo{year}{1990}\natexlab{}.
\newblock \showarticletitle{Exploiting fast matrix multiplication within the
  level 3 BLAS}.
\newblock \bibinfo{journal}{\emph{ACM Trans. Math. Softw.}}
  \bibinfo{volume}{16}, \bibinfo{number}{4} (\bibinfo{year}{1990}),
  \bibinfo{pages}{352–368}.
\newblock


\bibitem[\protect\citeauthoryear{Hunold, Rauber, and Runger}{Hunold
  et~al\mbox{.}}{2008}]%
        {Hunold-08}
\bibfield{author}{\bibinfo{person}{S. Hunold}, \bibinfo{person}{T. Rauber},
  {and} \bibinfo{person}{G. Runger}.} \bibinfo{year}{2008}\natexlab{}.
\newblock \showarticletitle{Combining building blocks for parallel multi--level
  matrix multiplication}.
\newblock \bibinfo{journal}{\emph{Parallel Comput.}}  \bibinfo{volume}{34}
  (\bibinfo{year}{2008}), \bibinfo{pages}{411--426}.
\newblock


\bibitem[\protect\citeauthoryear{Huss-Lederman, Jacobson, Tsao, Turnbull, and
  Johnson}{Huss-Lederman et~al\mbox{.}}{1996}]%
        {Huss-Super96}
\bibfield{author}{\bibinfo{person}{S. Huss-Lederman}, \bibinfo{person}{E.M.
  Jacobson}, \bibinfo{person}{A. Tsao}, \bibinfo{person}{T. Turnbull}, {and}
  \bibinfo{person}{J.R. Johnson}.} \bibinfo{year}{1996}\natexlab{}.
\newblock \showarticletitle{Implementation of Strassen’s algorithm for matrix
  multiplication}. In \bibinfo{booktitle}{\emph{Proc. ACM/IEEE Conf. on
  Supercomputing}}.
\newblock


\bibitem[\protect\citeauthoryear{Jia-Wei and Kung}{Jia-Wei and Kung}{1981}]%
        {10.1145/800076.802486}
\bibfield{author}{\bibinfo{person}{H. Jia-Wei} {and} \bibinfo{person}{H.~T.
  Kung}.} \bibinfo{year}{1981}\natexlab{}.
\newblock \showarticletitle{I/O Complexity: The Red-Blue Pebble Game}. In
  \bibinfo{booktitle}{\emph{Proc. ACM Symp. Theory of Computing (STOC)}}
  (Milwaukee, Wisconsin, USA). \bibinfo{publisher}{ACM},
  \bibinfo{pages}{326–333}.
\newblock
\showISBNx{9781450373920}


\bibitem[\protect\citeauthoryear{Kadhum, Qasem, Sleit, and Sharieh}{Kadhum
  et~al\mbox{.}}{2017}]%
        {kadhum2017efficient}
\bibfield{author}{\bibinfo{person}{M. Kadhum}, \bibinfo{person}{M.~H. Qasem},
  \bibinfo{person}{A. Sleit}, {and} \bibinfo{person}{A. Sharieh}.}
  \bibinfo{year}{2017}\natexlab{}.
\newblock \showarticletitle{Efficient MapReduce matrix multiplication with
  optimized mapper set}. In \bibinfo{booktitle}{\emph{Computer Science On-line
  Conference}}. Springer, \bibinfo{pages}{186--196}.
\newblock


\bibitem[\protect\citeauthoryear{K{\aa}gstr{\"o}m}{K{\aa}gstr{\"o}m}{2004}]%
        {kaagstrom2004management}
\bibfield{author}{\bibinfo{person}{Bo K{\aa}gstr{\"o}m}.}
  \bibinfo{year}{2004}\natexlab{}.
\newblock \showarticletitle{Management of deep memory hierarchies--recursive
  blocked algorithms and hybrid data structures for dense matrix computations}.
  In \bibinfo{booktitle}{\emph{Int. Workshop on Applied Parallel Computing}}.
  Springer, \bibinfo{pages}{21--32}.
\newblock


\bibitem[\protect\citeauthoryear{Kwasniewski, Kabi\'{c}, Besta, VandeVondele,
  Solc\`{a}, and Hoefler}{Kwasniewski et~al\mbox{.}}{2019}]%
        {COSMA}
\bibfield{author}{\bibinfo{person}{G. Kwasniewski}, \bibinfo{person}{M.
  Kabi\'{c}}, \bibinfo{person}{M. Besta}, \bibinfo{person}{J. VandeVondele},
  \bibinfo{person}{R. Solc\`{a}}, {and} \bibinfo{person}{T. Hoefler}.}
  \bibinfo{year}{2019}\natexlab{}.
\newblock \showarticletitle{Red-Blue Pebbling Revisited: Near Optimal Parallel
  Matrix-Matrix Multiplication}. In \bibinfo{booktitle}{\emph{Proc. Int. Conf.
  High Performance Computing, Networking, Storage and Analysis}}
  \emph{(\bibinfo{series}{SC '19})}. Article \bibinfo{articleno}{24},
  \bibinfo{numpages}{22}~pages.
\newblock
\showISBNx{9781450362290}
\urldef\tempurl%
\url{https://doi.org/10.1145/3295500.3356181}
\showDOI{\tempurl}


\bibitem[\protect\citeauthoryear{Luo and Drake}{Luo and Drake}{1995}]%
        {Luo-SAC95}
\bibfield{author}{\bibinfo{person}{Q. Luo} {and} \bibinfo{person}{J. Drake}.}
  \bibinfo{year}{1995}\natexlab{}.
\newblock \showarticletitle{A scalable parallel Strassen's matrix
  multiplication algorithm for distributed-memory computers}. In
  \bibinfo{booktitle}{\emph{Proc. ACM Symp. Applied Computing, SAC'95}}.
  \bibinfo{pages}{221--226}.
\newblock


\bibitem[\protect\citeauthoryear{Peise and Bientinesi}{Peise and
  Bientinesi}{2017}]%
        {peise2017algorithm}
\bibfield{author}{\bibinfo{person}{E. Peise} {and} \bibinfo{person}{P.
  Bientinesi}.} \bibinfo{year}{2017}\natexlab{}.
\newblock \showarticletitle{Algorithm 979: recursive algorithms for dense
  linear algebra—the ReLAPACK collection}.
\newblock \bibinfo{journal}{\emph{ACM Trans. Math. Softw. (TOMS)}}
  \bibinfo{volume}{44}, \bibinfo{number}{2} (\bibinfo{year}{2017}),
  \bibinfo{pages}{1--19}.
\newblock


\bibitem[\protect\citeauthoryear{Qasem, Sarhan, Qaddoura, and Mahafzah}{Qasem
  et~al\mbox{.}}{2017}]%
        {qasem2017matrix}
\bibfield{author}{\bibinfo{person}{M.~H. Qasem}, \bibinfo{person}{A.~A.
  Sarhan}, \bibinfo{person}{R. Qaddoura}, {and} \bibinfo{person}{B.~A.
  Mahafzah}.} \bibinfo{year}{2017}\natexlab{}.
\newblock \showarticletitle{Matrix multiplication of big data using mapreduce:
  a review}. In \bibinfo{booktitle}{\emph{2nd Int. Conf. Applications of
  Information Technology in Developing Renewable Energy Processes \& Systems
  (IT-DREPS)}}. IEEE, \bibinfo{pages}{1--6}.
\newblock


\bibitem[\protect\citeauthoryear{Song, Dongarra, and Moore}{Song
  et~al\mbox{.}}{2006}]%
        {Song-PDCS06}
\bibfield{author}{\bibinfo{person}{F. Song}, \bibinfo{person}{J. Dongarra},
  {and} \bibinfo{person}{S. Moore}.} \bibinfo{year}{2006}\natexlab{}.
\newblock \showarticletitle{Experiments with Strassen's algorithm: From
  sequential to parallel}. In \bibinfo{booktitle}{\emph{Proc. Parallel and
  Distributed Computing and Systems, (PDCS)}}.
\newblock


\bibitem[\protect\citeauthoryear{Stothers}{Stothers}{2003}]%
        {stothers2010complexity}
\bibfield{author}{\bibinfo{person}{A.J. Stothers}.}
  \bibinfo{year}{2003}\natexlab{}.
\newblock \showarticletitle{On the complexity of matrix multiplication}.
\newblock \bibinfo{journal}{\emph{Journal of Complexity}}  \bibinfo{volume}{19}
  (\bibinfo{year}{2003}), \bibinfo{pages}{43--60}.
\newblock
Issue 1.


\bibitem[\protect\citeauthoryear{Strang}{Strang}{2006}]%
        {strang06}
\bibfield{author}{\bibinfo{person}{G. Strang}.}
  \bibinfo{year}{2006}\natexlab{}.
\newblock \bibinfo{booktitle}{\emph{Linear Algebra and Its Applications, Fourth
  Ed.}}
\newblock \bibinfo{publisher}{Thomson Brooks/Cole}.
\newblock


\bibitem[\protect\citeauthoryear{Strassen}{Strassen}{1969}]%
        {strassen1969gaussian}
\bibfield{author}{\bibinfo{person}{V. Strassen}.}
  \bibinfo{year}{1969}\natexlab{}.
\newblock \showarticletitle{Gaussian elimination is not optimal}.
\newblock \bibinfo{journal}{\emph{Numerische mathematik}} \bibinfo{volume}{13},
  \bibinfo{number}{4} (\bibinfo{year}{1969}), \bibinfo{pages}{354--356}.
\newblock


\bibitem[\protect\citeauthoryear{Thottethodi, Chatterjee, and
  Lebeck}{Thottethodi et~al\mbox{.}}{1998}]%
        {thottethodi1998tuning}
\bibfield{author}{\bibinfo{person}{M. Thottethodi}, \bibinfo{person}{S.
  Chatterjee}, {and} \bibinfo{person}{A.R. Lebeck}.}
  \bibinfo{year}{1998}\natexlab{}.
\newblock \showarticletitle{Tuning Strassen's matrix multiplication for memory
  efficiency}. In \bibinfo{booktitle}{\emph{Proc. 1998 ACM/IEEE Conf. on
  Supercomputing (SC'98)}}. IEEE, \bibinfo{pages}{36--36}.
\newblock


\bibitem[\protect\citeauthoryear{Wang, Zhang, Zhang, Lu, Wu, and Wang}{Wang
  et~al\mbox{.}}{2014}]%
        {wang2014intel}
\bibfield{author}{\bibinfo{person}{E. Wang}, \bibinfo{person}{Q. Zhang},
  \bibinfo{person}{B.~Shenand~G. Zhang}, \bibinfo{person}{X. Lu},
  \bibinfo{person}{Q. Wu}, {and} \bibinfo{person}{Y. Wang}.}
  \bibinfo{year}{2014}\natexlab{}.
\newblock \showarticletitle{Intel math kernel library}.
\newblock In \bibinfo{booktitle}{\emph{High-Performance Computing on the
  Intel{\textregistered} Xeon Phi™}}. \bibinfo{publisher}{Springer},
  \bibinfo{pages}{167--188}.
\newblock


\bibitem[\protect\citeauthoryear{Williams}{Williams}{2012}]%
        {williams2012multiplying}
\bibfield{author}{\bibinfo{person}{V.V. Williams}.}
  \bibinfo{year}{2012}\natexlab{}.
\newblock \showarticletitle{Multiplying matrices faster than
  {C}oppersmith-{W}inograd}. In \bibinfo{booktitle}{\emph{Proc. 44th ACM Symp.
  Theory of Computing (STOC)}}. \bibinfo{pages}{887–898}.
\newblock


\bibitem[\protect\citeauthoryear{Yang and Miller}{Yang and Miller}{1988}]%
        {yang1988critical}
\bibfield{author}{\bibinfo{person}{C-Q. Yang} {and} \bibinfo{person}{B.~P.
  Miller}.} \bibinfo{year}{1988}\natexlab{}.
\newblock \showarticletitle{Critical path analysis for the execution of
  parallel and distributed programs}. In \bibinfo{booktitle}{\emph{Proc. 8th
  Int. Conf. on Distributed Computing Systems (ICDCS)}}. IEEE,
  \bibinfo{pages}{366--373}.
\newblock


\bibitem[\protect\citeauthoryear{Zeng, Guo, Luo, and Gu}{Zeng
  et~al\mbox{.}}{2012}]%
        {zeng2012discrete}
\bibfield{author}{\bibinfo{person}{W. Zeng}, \bibinfo{person}{R. Guo},
  \bibinfo{person}{F. Luo}, {and} \bibinfo{person}{X. Gu}.}
  \bibinfo{year}{2012}\natexlab{}.
\newblock \showarticletitle{Discrete heat kernel determines discrete Riemannian
  metric}.
\newblock \bibinfo{journal}{\emph{Graphical Models}} \bibinfo{volume}{74},
  \bibinfo{number}{4} (\bibinfo{year}{2012}), \bibinfo{pages}{121--129}.
\newblock


\end{thebibliography}

\end{document}